\documentclass{sig-alternate-2013-modified}
\usepackage{amsmath,mathtools}
\usepackage{graphicx}
\usepackage{mathptmx}
\usepackage{algorithm}
\usepackage[algo2e, ruled, vlined]{algorithm2e}
\usepackage[pdftex]{xcolor}
\usepackage{url}
\usepackage{microtype}
\usepackage{textcomp}
\usepackage{booktabs}
\usepackage{wrapfig}
\newcommand{\increase}[2]{\ensuremath{#1 \, {\overset{+}{\gets}}\,
    #2}} 
\newcommand{\decrease}[2]{\ensuremath{#1 \, {\overset{-}{\gets}}\, #2}} 

\usepackage{sidecap}
\usepackage{epsfig}
 \usepackage{ifpdf}
 \ifpdf\fi
\usepackage[sort,compress]{cite}

\newtheorem{thm}{Theorem}[section]
\newtheorem{theorem}{Theorem}[section]

\newtheorem{lemma}[thm]{Lemma}

\newcommand{\ignore}[1]{} 
\newcommand{\notinproc}[1]{#1}
\newcommand{\onlyinproc}[1]{}

\newcommand{\notinarxiv}[1]{}
\newcommand{\onlyinarxiv}[1]{#1}

\newcommand\E{\textsf{E}}

\newcommand{\kth}{\text{k}^{\text{th}}}

\newcommand{\ADS}{\mathop{\rm ADS}}
\newcommand{\INF}{\mathop{\sf Inf}}
\newcommand{\cADS}{\mathop{\rm cADS}}
\newcommand{\cascade}[1]{G^{(#1)}}
\newcommand{\lstar}{\ensuremath{\mbox{L}^*}}
\newcommand{\instance}[1]{\textsf{#1}}
\newcommand{\skim}{{\sc SKIM}}
\newcommand{\alphaskim}{{\sc $\alpha$-SKIM}}
\newcommand{\askim}{{\sc $\alpha$-SKIM}}
\newcommand{\tskim}{{\sc $T$-SKIM}}
\SetAlgoCaptionSeparator{:} 
\DeclareMathOperator*{\argmax}{argmax}

\SetKwFunction{MargGain}{MargGain}
\SetKwArray{cdelta}{$\delta$}
\SetKwFunction{AddSeed}{AddSeed}
\SetKw{True}{true}
\SetKw{False}{false}
\SetKw{Continue}{continue}
\SetKw{Skip}{skip}
\SetKw{Prune}{prune}
\SetKw{Return}{return}
\SetKw{Or}{or}
\SetKwFunction{Append}{append}
\SetKwArray{Covered}{$\delta$}
\SetKwArray{Index}{index}
\SetKwArray{Size}{size}
\SetKwArray{SeedList}{seedlist}
\SetKwData{Rank}{rank}
\SetKwData{Zeroranks}{Z}
\SetKwData{Queue}{Q}
\SetKwFunction{Insert}{insert}
\SetKwFunction{Erase}{erase}
\SetKwFunction{Append}{append}
\SetKwFunction{Maxelement}{max\_element}
\SetKwFunction{Deletemax}{delete\_max}
\SetKwFunction{Lstar}{L\textsuperscript{*}}
\SetKwFunction{Sort}{sort}
\SetKwArray{Skylines}{skylines}
\SetKw{Skip}{skip}
\SetKw{Return}{return}
\SetKwFunction{Pushback}{append}
\SetKwFunction{Insert}{insert}
\SetKwFunction{Deletemax}{delete\_max}
\SetKwFunction{Maxelement}{max\_element}
\SetKwFunction{Merge}{merge}
\SetKwData{Ads}{ADS}
\SetKwData{Tempsketch}{tempsketch}
\SetKwData{Numzero}{numzero}
\SetKwData{Queue}{Q}
\SetKw{Return}{return}
\SetKwFunction{NextSeed}{NextSeed} 
\SetKwFunction{Deletemax}{delete\_max}
\SetKwFunction{Maxval}{max\_priority}
\SetKwFunction{Maxelement}{max\_element}

\newcommand{\comment}[1]{\footnote{#1}}

\definecolor{darkgreen}{RGB}{0,100,0}
\definecolor{orange}{RGB}{255,80,0}
\newcommand{\edith}[1]{\comment{\textcolor{purple}{EC}:#1}}
\newcommand{\daniel}[1]{\comment{\textcolor{darkgreen}{DD}:#1}}

\newcommand{\Xcomment}[1]{}

\title{Distance-Based Influence in Networks: \\ Computation and
  Maximization\onlyinproc{ \\{\large KDD 2016 submission \#119}}}

\ignore{
\numberofauthors{1}
\author{\alignauthor Edith Cohen \\
{\large Tel Aviv University}\\
{\large Google Research}\\
{\large edith@cohenwang.com}
\alignauthor Daniel Delling\\
{\large Sunnyvale, CA, USA}\\
{\large daniel.delling@gmail.com}
\alignauthor Thomas Pajor \\
{\large Microsoft Research}\\
{\large USA}\\
{\large thomas@tpajor.com}
\alignauthor Renato F. Werneck\\
{\large San Francisco, CA, USA}\\
{\large rwerneck@gmail.com}
}
}

\author{
\begin{tabular}{cccc}
Edith Cohen  & Daniel Delling & Thomas Pajor  & Renato F. Werneck \\
{\large Tel Aviv University} & {\large Sunnyvale} & {\large  Cupertino} &   {\large San Francisco} \\ 
{\large Google Research} &  {\large USA} & {\large USA} &  {\large USA} \\
{\large edith@cohenwang.com} & {\large daniel.delling@gmail.com} &
  {\large  thomas@tpajor.com}  & {\large rwerneck@acm.org}
\end{tabular}
}

\ignore{
 \numberofauthors{1}
\author{
\alignauthor Edith Cohen,  Daniel Delling,  Thomas
Pajor, Renato F. Werneck\\
       { \email{edith@cohenwang.com} (Tel Aviv University),
         \email{daniel.delling@gmail.com} (Sunnyvale, USA),
         \email{tpajor@microsoft.com} (Microsoft Research),
         \email{rwerneck@cs.princeton.edu} (San Francisco, USA)}
}
}

\begin{document}
\maketitle

\begin{abstract}

  A premise at a heart of network analysis is that entities in a
network derive utilities from their connections.  The {\em
influence} of a seed set $S$ of nodes is defined as the sum over
nodes $u$ of the {\em utility} of $S$ to $u$.  {\em Distance-based}
utility, which is a decreasing function of the distance from $S$ to
$u$, was explored in several successful research threads from social
network analysis and economics: Network formation games [Bloch and
Jackson 2007], Reachability-based influence [Richardson and Domingos
2002; Kempe et al. 2003]; ``threshold'' influence [Gomez-Rodriguez
et al. 2011]; and {\em closeness centrality} [Bavelas 1948].

We formulate a model that unifies and extends this previous work and
address the two fundamental computational problems in this domain:
{\em Influence oracles} and {\em influence maximization} (IM).  An
oracle performs some preprocessing, after which influence queries for
arbitrary seed sets can be efficiently computed.  With IM, we seek a
set of nodes of a given size with maximum influence.  Since the IM
problem is computationally hard, we instead seek a {\em greedy
  sequence} of nodes, with each prefix having influence that is at
least $1-1/e$ of that of the optimal seed set of the same size. We
present the first highly scalable algorithms for both problems,
providing statistical guarantees on approximation quality and
near-linear worst-case bounds on the computation. We perform an
experimental evaluation which demonstrates the effectiveness of our
designs on networks with hundreds of millions of edges.

\end{abstract}

\section{introduction}



 Structural notions of the {\em influence} of a set of entities in a
 network which are based on the
{\em utility} an entity derives from its connectivity to others,  are central
 to network analysis and were studied in the context of
 social and economic models
 \cite{Bavelas:HumanOrg1948,Sabidussi:psychometrika1966,Freeman:sn1979,KKT:KDD2003,BlochJackson:2007,CoKa:jcss07,JacksonNetworks:Book2010,Opsahl:2010}
 with applications that include ranking, covering, clustering, and understanding diffusion and
  network formation games.  
More formally, influence is typically defined in terms of
utilities $u_{ij}$ between ordered pairs of entities.
 The influence of a single entity $i$, also known as its
{\em centrality,} is the sum of utilities $\INF(i) = \sum_{j} u_{ij}$ over 
entities $j$. The influence of a set $S$ of entities is the sum over
entities of the highest utility match from $S$:
$$\INF(S) = \sum_j \max_{i\in S} u_{ij}\ .$$

  One of the simplest and more popular definitions of influence relies on
{\em reachability-based} utility
\cite{GLM:marketing2001,RichardsonDomingos:KDD2002,KKT:KDD2003}.
The network here is a directed graph, where entities correspond to
nodes. We have $u_{ij}=1$ when the node $j$ is
reachable from node $i$.  Therefore, the influence of $S$ is the number of
entities reachable from $S$.
A powerful enhancement of this model is to 
allow for multiple  {\em
  instances}, where each instance is a set of directed
edges, or for a distribution over
instances, and accordingly, define $u_{ij}$ as the respective average or
expectation.
The popular Independent Cascade (IC)
model of Kempe et al.~\cite{KKT:KDD2003} uses a distribution defined
by a graph with independent inclusion
probabilities for edges.

More expressive utility is based on
 shortest-path distances
 \cite{Bavelas:HumanOrg1948,Sabidussi:psychometrika1966,Freeman:sn1979,KKT:KDD2003,BlochJackson:2007,CoKa:jcss07,Opsahl:2010}. Specifically, 
the utility is $u_{ij} = \alpha(d_{ij})$, where $\alpha(d) \geq 0$ is a non-increasing function
\cite{Rosenblatt:stats1956,CoSt:pods03f,BlochJackson:2007,CoKa:jcss07} 
applied to the  respective shortest-path distance.
Reachability-based utility falls out as a special case, using 
$\alpha(x)=1$ for finite $x$ and $\alpha(+\infty)=0$.  
Popular kernels include 
exponential decay $\alpha(x)=\exp(-\lambda x)$,
polynomial decay $\alpha(x)=1/\mbox{poly}(x)$, Gaussian 
$\alpha(x)=\exp(-\lambda x^2)$, and threshold, which is obtained using $\alpha(x)=1$ when 
$x\leq T$ and $\alpha(x)=0$ otherwise.  
This variety, used in practice, demonstrates the value of this modeling 
flexibility.

A distance-based model for information diffusion in social networks was recently proposed by
Gomez-Rodriguez et
al.~\cite{Gomez-RodriguezBS:ICML2011}. In their model, the seed set $S$
corresponds to the ``infected'' nodes and edge lengths correspond
to propagation time.  The shortest paths distance from $S$ to a node
$v$ corresponds to the elapsed time until $v$ is
``infected.''   The decay function models the amount by which slower propagaing is
less valuable.
Again, the modelling power is enhanced by working with
multiple instances (each instance is a set of directed edges with
lengths), or with a distribution over such instances, and 
defining the utility as the average of $\alpha(d_{ij})$ over instances
or as the expectation $u_{ij} = \E[\alpha(d_{ij})]\ .$
In particular, a natural model associates independent 
randomized edge lengths (REL) with live
edges~\cite{Gomez-RodriguezBS:ICML2011,ACKP:KDD2013,CDFGGW:COSN2013,DSGZ:nips2013},
using exponential~\cite{Gomez-RodriguezBS:ICML2011,ACKP:KDD2013,CDFGGW:COSN2013} or
Weibull~\cite{DSGZ:nips2013} distributions.

We provide an intuitive explanation to the power of randomization to improve the quality of
reachability-based and distance-based utility measures.
We expect the utility $u_{ij}$ to be higher when there is a short path
and also to when there are more paths.  The deterministic shortest-path
distance does not reflect paths multiplicity, but with REL, a pair
connected by more paths will have a shorter expected distance  than
another pair, even when distances according to expected edge
lengths are equal~\cite{CDFGGW:COSN2013}. 
In addition, randomization, even as introduction of random noise,
 is a powerful tool in learning, as it 
reduces the sensitivity of the results to 
insignificant variations in the input and overfitting the data.

We note here that the work of Gomez-Rodriguez et al. and subsequent
  work on distance-based diffusion~\cite{Gomez-RodriguezBS:ICML2011,DSGZ:nips2013} was focused
  specifically on
  threshold decay functions, where for a threshold parameter $T$,
  $u_{ij}=1$ only when the distance is at most $T$.
Distance-based utility with smooth decay functions, which we study here, naturally occur
  in the physical world and was extensively studied also in the
  context of data analysis \cite{Rosenblatt:stats1956} and 
networks \cite{CoKa:jcss07,BlochJackson:2007,Opsahl:2010,JacksonNetworks:Book2010}.
In particular, distance-based influence generalizes the distance-decaying variant 
\cite{CoKa:jcss07,Dangalchev:2006,Opsahl:2010,BoldiVigna:IM2014,ECohenADS:PODS2014} of 
closeness centrality \cite{Bavelas:HumanOrg1948}, which was studied with 
exponential, harmonic, threshold, and general decay functions.


The two fundamental algorithmic problems in applications of influence are
{\em influence computation} and {\em influence 
    maximization} (IM). 

Influence computation is the problem of computing the influence of a
specified seed set $S$ of nodes.  This can be done using multiple
single-source shortest-paths computations from the seed set, but the
computation does not scale well when there are many queries on very
large networks.  Cohen et al.~for reachability-based influence~\cite{binaryinfluence:CIKM2014}
and Du et al.~for threshold influence~\cite{DSGZ:nips2013} designed {\em
 influence oracles} which preprocess the input so that influence queries 
$\INF(\mathcal{G},S)$ for a specified 
seed set $S$ can be approximated quickly.  
In both cases, a sketch, based on~\cite{ECohen6f}, is computed for each node
so that the influence of a seed set $S$ can be estimated from the 
sketches of the nodes in $S$.  
\notinproc{
For general distance-based influence, 
we can consider an oracle designed for a pre-specified function 
$\alpha$ (say reachability-based or threshold), or a more powerful oracle which
allows $\alpha$ to
be specified at query time.}

Influence maximization is the problem of 
finding a seed set $S\subset V$ with maximum 
 influence, where $|S| = s$ is given.  
%
Since reachability-based and threshold influence~\cite{KKT:KDD2003,Gomez-RodriguezBS:ICML2011} are special cases,
we know that distance-based influence maximization with general $\alpha$ 
is NP-complete and hard to approximate to anything better than $1-(1-1/s)^s$ of the optimum
  for a seed set of size $s$ (the hardness result  is asymptotic in
  $s$)~\cite{feige98}.
Fortunately, from
monotonicity and submodularity of these influence
functions~\cite{KKT:KDD2003,Gomez-RodriguezBS:ICML2011}, we 
obtain that the greedy algorithm ({\sc Greedy}), which
iteratively adds to the seed set the node with maximum marginal
influence, is guaranteed to provide a solution that is at least
$1-(1-1/s)^s > 1-1/e$ of
the optimum \cite{submodularGreedy:1978}.  This (worst-case) guarantee holds for every prefix size of the
sequence of seeds reported, which means that the {\sc Greedy} sequence
approximates the Pareto
front of the  trade-off of seed
 set size versus its influence.  The Pareto front provides added value
 since it characterizes the
 influence coverage of the network and can be used to find
bi-criteria sweet spots between the size of the seed set and its coverage.

In terms of solution quality, also in practice, {\sc Greedy} had been the gold standard
for submodular maximization.  Scalability, however, remains an
issue for reachability-based influence even with various optimizations
\cite{Leskovec:KDD2007,CWY:KDD2009}, and even when working with a single
instance.
As a result, extensive research work on reachability-based influence maximization 
proposes scalable heuristics \cite{JHC:ICDM2012},
scalability with guarantees only for a small number of seeds
\cite{BBCL:SODA2014,TXS:sigmod2014}, and, more recently, \skim, which
computes a full approximate {\sc Greedy}
sequence~\cite{binaryinfluence:CIKM2014} while
scaling nearly linearly with input size.
For threshold influence, existing maximization
algorithms are by Du et al.~\cite{DSGZ:nips2013}  based on
sketches~\cite{ECohen6f} and by \cite{TSX:sigmod2015} based on an
extension of the reachability-based algorithm of
\cite{TXS:sigmod2014}.  
These designs, when using near-linear computation, can only provide
approximation guarantees for a small number of seeds, even when edge
lengths are not randomized.\notinproc{\footnote{Note that \skim, which is the design we build on here, provides strong
  worst-case statistical guarantees on both accuracy and running time,
for producing a full approximate  greedy sequence of seeds.  The
computation is proportional to the size of the instances and is
near-linear when the number of instances is small. Other existing 
algorithms, even related ones also based on reverse searches
  \cite{ECohen6f,DSGZ:nips2013 ,BBCL:SODA2014,TXS:sigmod2014}, do
not compute a full sequence in near-linear time: The
  running time grows with the number of seeds even on a single
  instance, which is the same as an IC  model with all probabilities being $0$ or $1$.
The inherent issue is the quality of the approximate greedy
  solution.  The performance of an approximate greedy algorithm on a sample 
deteriorates quickly with the number of seeds, so a very large sample
size, $O(sn)$ in total, may be needed when we are interested in good
approximation with respect to all seed sets of size $s$, which is
needed in order to correctly identify the maximum one.
The power and novelty of \skim\ comes from 
efficiently maintaining a {\em residual problem}, which allows for
samples which provide relative error guarantees on {\em marginal contributions} of  new seed
candidates.  We also note that \skim\ runs in  worst-case near-linear
time on inputs specified as a set of (arbitrary) instances.  It is not known
if near-linear time algorithms that compute an approximate greedy
sequence exist for the IC model.
}}

\medskip
\noindent
{\bf  Contributions.}
Our contributions are as follows.
Our distance-based influence model is presented in
Section \ref{model:sec}. 
We define exact and
approximate  {\sc Greedy} sequences and establish approximation
guarantees of both.  We then formulate a representation of a
{\em residual problem} with respect to a seed set $S$, which
facilitates the greedy computation of the next seed.  Finally, we establish 
a probabilistic bound on residual updates, which magically holds for
the approximate but not for the exact greedy.  This bound is a
critical component in obtaining near-linear computation of a full
approximate greedy sequence.

In Section~\ref{thresh:sec} we consider the threshold influence model~\cite{Gomez-RodriguezBS:ICML2011,DSGZ:nips2013}.  We extend the  (approximate) influence oracles and the 
\skim\ reachability-based influence maximization algorithm \cite{binaryinfluence:CIKM2014},
to obtain oracles and  \tskim\ for threshold influence. 
The extension itself replaces reachability searches with pruned
Dijkstra computations.
The statistical guarantees on approximation quality are inherited from
\skim.  The worst-case 
running time analysis, establishing that \tskim\ computes the full 
approximate greedy sequence in near-linear time,  required our probabilistic bound on residual
updates.
\ignore{
We show that our algorithms scale much better  than the state-of-the-art threshold influence algorithms 
of Du et al.~\cite{DSGZ:nips2013}.  In particular, a much smaller
storage (sketch size) is needed for obtaining an oracle with the same
accuracy guarantees, and IM is orders of magnitude faster.
}

In Section~\ref{timedQ:sec}, we present {\em distance-based influence
oracles}, which take as input a seed
set $S$ and ({\em any}) decay function $\alpha$, which can be specified
at query time.  As explained earlier, many different decay/kernel
functions are used extensively in practice, which makes the
flexibility of the oracle to handle arbitrary $\alpha$ valuable.
Our oracle computes a novel sketch for each node;  the 
{\em combined All-Distances sketch} ($\cADS$), which generalizes
All-Distances Sketches
(ADS)~\cite{ECohen6f,CoKa:jcss07,bottomk07:ds,ECohenADS:PODS2014,BoldiVignaHyperball:arxiv2014}, used for
closeness centrality computation,  to 
multiple instances or probabilistic models.  These per-node sketches
have expected size at most $k\ln (n\min\{k,\ell\})$ (with good concentration), where $n$ is the number of nodes, $\ell$ is the number of instances, and $k$ is a sketch parameter that determines a trade-off between information 
and the amounts of computation and space required.
We estimate the distance-based influence of a seed set $S$ from the sketches of the nodes in $S$. Our estimator
uses  HIP probabilities~\cite{ECohenADS:PODS2014} with the 
L$^*$ estimator~\cite{sorder:PODC2014}, which optimally uses the
information in the sketches, and has worst-case
coefficient of variation (CV) $\leq 1/\sqrt{2k-2}$.

In Section~\ref{timedIM:sec} we present \alphaskim, the first scalable
influence maximization algorithm that applies with general decay functions $\alpha$.
Our design is a strong contribution from both theoretical and practical
perspectives, providing a novel near-linear worst-case bound on the
computation, performance guarantees that nearly match those of exact
{\sc Greedy}, and a scalable implementation
that runs on large networks.  The heart of our design is a novel algorithmic technique
of efficiently maintaining weighted samples from influence sets that
allows us to accurately estimate marginal influences as nodes are
added to the seed set.


Section~\ref{sec:exp} presents a comprehensive experimental evaluation
of our algorithms. For threshold influence oracles
and maximization, we obtain three orders of magnitude speedups over the algorithms of Du et
al.~\cite{DSGZ:nips2013}, with no loss in quality.  Even though both approaches apply the
sketches of Cohen~\cite{ECohen6f}, we are able to obtain improvements by
working with combined sketches, applying better estimators, and, in our
IM algorithm, only computing sketches to the point needed to determine the next node. 
We also show that the generalization to arbitrary decay functions $\alpha$ is only slightly slower, and can easily handle graphs with hundreds of millions of edges.

\notinproc{We
conclude in Section \ref{conclu:sec}.} \onlyinproc{Due to page limits,
  most proofs and many details are omitted.  We refer the reader to 
a full version in http://arxiv.org/abs/1410.6976
}

\section{Distance-Based Model}  \label{model:sec}

 For a set of entities $V$, {\em utility} $u_{vu}$ for $v,u\in V$, and a
 {\em seed set} $S \subset V$, we
 define the {\em influence} of $S$ as 
$$\INF(S) = \sum_{u\in V} \min_{v\in S} u_{uv}\ .$$
{\em Distance-based} utility is defined with respect to
a  non-increasing function $\alpha$ such that
$\alpha(\infty)\equiv 0$
and (a set or distribution over)
edge-weighted graphs $G=(V,E,w)$, where the nodes correspond to
entities
and edges $e\in E$ have lengths $w(e)>0$.

We refer to a single graph $G$ as an {\em instance}.
We denote by $d_{vu}$ the shortest-path distance in $G$ from $v$ to $u$. 
When there is no path from $v$ to $u$ in $G$ we define $d_{vu}\equiv
\infty$.  
For a set of nodes $S$, we let 
$d_{Su}= \min_{v\in S} d_{vu}$ be the shortest-path distance in $G$ from $S$
to $u$.

The utility with respect to a single instance is defined as $u_{vu} =
\alpha(d_{vu})$, yielding the influence function
$$\INF(G,S) = \sum_{u\in V} \alpha(d_{Su}).$$


\ignore{
The well-studied special case of
{\em reachability-based} influence is obtained when using uniform
edge lengths $w(e)\equiv 1$ and $\alpha(x)\equiv 1 
 \iff x<\infty$; the reachability-based influence of $S$ is the cardinality of its reachability set, that is,
the number of nodes reachable from at least one node in
$S$.}


For a set $\mathcal{G}=\{\cascade{i}\}$
of $\ell\geq 1$  instances 
$\cascade{i}=(V,E^{(i)},w^{(i)})$, we define the 
utility as the average $u_{vu} = \frac{1}{\ell}
\sum_{i=1}^\ell \alpha(d^{(i)}_{vu})$ and accordingly the influence is
the average of the single-instance influence:
\begin{equation} \label{timedinf}
\INF(\mathcal{G},S)=\INF(\{\cascade{i}\},S) = \frac{1}{\ell} \sum_{i\in [\ell]}
\INF(\cascade{i},S). 
\end{equation}
Our algorithms work with inputs specified as a set of one or more instances and
with respect to the influence function \eqref{timedinf}.

A set of instances can be 
derived from traces or generated by Monte-Carlo simulations of 
a probabilistic model $\mathcal{G}$. Such a model defines a
distribution
 over instances $G \sim \mathcal{G}$ which 
share a set $V$ of nodes. The respective utility is then 
$u_{uv} =
\E[\alpha(d_{uv})]$ and the 
influence of $\mathcal{G}$ is then the expectation 
\begin{equation} \label{timedprobinf}
\INF(\mathcal{G},S)=\E_{G\sim \mathcal{G}} \INF(G,S). 
\end{equation}
When $\INF$ is well concentrated around its expectation, a small
number of simulations suffices to approximate the results.



\subsection{The Exact Greedy Sequence} \label{greedy:sec}
 
 We present the exact greedy algorithm for the distance-based influence objective 
 $\INF(\{\cascade{i}\},S)$, as defined in Equation~\eqref{timedinf}.
{\sc Greedy} 
starts with an empty seed set $S=\emptyset$.  In each iteration,
it adds the node with maximum marginal gain, that is, the node $u$ 
that maximizes the influence of $S\cup\{u\}$. 
\ignore{
\begin{algorithm} \caption{{\sc Greedy}}
\begin{algorithmic}
\Statex $S\gets \emptyset$ \Comment{initialize seed set}
\While {$|S|< s$}
\State $c \gets \arg\max_{i\not\in S}  \INF(\{\cascade{i}\},S\cup\{i\})$
\State $S \gets S \cup \{c\}$
\EndWhile 
\end{algorithmic}
\label{basicgreedy:alg}
\end{algorithm}
}
{\sc Greedy} thus produces  a sequence of nodes, providing an approximation
 guarantee for the seed set defined by each prefix.  

 We now elaborate on the computation of the marginal gain of $u$ given 
 $S$.  To do so efficiently as $S$ grows, we work with a {\em residual
   problem}.  
We denote the residual problem of $\mathcal{G}$
with respect to seed set $S$ as $\mathcal{G}|S$.
The influence of a set of nodes $U$ in the residual problem
is equal to the marginal influence in the original problem:
$$\INF(\mathcal{G}|S,U) = 
\INF(\mathcal{G}, S\cup U)-\INF(\mathcal{G}, S).$$

  A residual problem has a slightly more general specification. The
  input has the form $(\mathcal{G},\delta)$, where $\delta^{(i)}_v \geq 0$ maps
  node-instance pairs $(v,i)$ to nonnegative numbers.
The $\delta$ values we use for the residual problem $\mathcal{G}|S$ are
the respective distances from the seed set (but can be truncated
without violating correctness 
at any 
distance $x$ in which $\alpha(x)=0$):
$$\delta^{(i)}_v = d^{(i)}_{Sv} \equiv \min_{u\in S} d^{(i)}_{uv}.$$
When the seed set is empty or the node $v$ is either not reachable
from $S$ in instance $i$ or has distance  $d^{(i)}_{Sv} > \sup_x
  \{\alpha(x)>0\}$, 
we can use $\delta^{(i)}_v = \infty$ or 
any $\delta^{(i)}_v > \sup_x \{\alpha(x)>0\}$, which is equivalent.

We now extend
the influence definition for inputs of the form $(\mathcal{G},\delta)$.
For a node $u$, we consider the contribution  of each node-instance pair
$(v,i)$ to the influence of $u$:
\begin{equation} \label{Deltadef:eq}
\Delta^{(i)}_{uv} \equiv \max\{0,
\alpha(d^{(i)}_{uv})-\alpha(\delta^{(i)}_v)\}.
\end{equation}
  The influence of $u$ in the residual problem is the (normalized) sum of these
  contributions over all nodes in all instances:
\begin{equation} \label{minfdef}
\INF((\mathcal{G},\delta),u) \equiv \frac{1}{\ell} \sum_i \sum_v
\Delta^{(i)}_{uv}.
\end{equation}
 It is not hard to verify the following.
\begin{lemma}
For any set of nodes $U$,
the influence of $U$ in $\mathcal{G}|S$ is the
 same as marginal influence of $U$ with respect to $S$ in
 $\mathcal{G}$.  
\end{lemma}

 Given a residual input $(\mathcal{G},\delta)$, the influence of a 
 node $u$ (which is the same as its marginal influence \onlyinproc{$\MargGain(u)$ }in the original 
 input $\mathcal{G}$) can be computed using a pruned application of 
 Dijkstra's algorithm from $u$.\notinproc{ A pseudocode is provided as the 
 function $\MargGain(u)$ in Appendix~\ref{app:ps:funcgreedy}.}
The pruning is performed for efficiency 
 reasons by avoiding expanding the search in futile directions.  In 
 particular, we can always prune at distance $d$ when $\alpha(d)=0$ or 
 when $d\geq \delta^{(i)}_v$.  The correctness of the pruning follows 
 by observing that all nodes $u$ Dijkstra could reach from the pruned 
 node have $\Delta^{(i)}_{uv}=0$. 

At each step, {\sc Greedy} selects a node with maximum influence in the residual input. 
It then updates the distances $\delta$ so that they 
capture the residual problem 
$\mathcal{G}| S\cup\{u\}$.  \onlyinproc{For details see the full version.}\notinproc{A pseudocode for 
updating $\delta$ is provided as 
the function $\AddSeed(u)$ in Appendix~\ref{app:ps:funcgreedy}.  To update, we perform a 
pruned single-source shortest-paths computation in each instance $i$ from $u$, as in 
$\MargGain(u)$.}


\notinproc{
A straightforward implementation of {\sc Greedy} will recompute $\MargGain(u)$ for  all 
nodes $u$ after each iteration.  A common acceleration is instead to perform 
lazy evaluations: one keeps older values of marginal gains, which are
upper bounds on the current values, and updates the current values
only for the candidates at the top of the queue as needed to determine
the maximum one.
Lazy evaluations for reachability-based influence 
  were used by Leskovec et al.~\cite{Leskovec:KDD2007} in their CELF algorithm. 
The correctness of lazy evaluations follows from 
the submodularity and monotonicity of the objective, which imply that the marginal gain of a node can only decrease as the seed set grows. 
}

\subsection{Approximate Greedy Sequences} \label{approxgreedy:sec}
{\sc Approximate Greedy} is similar to exact {\sc Greedy}, but
in each iteration, instead of selecting a seed node with maximum
marginal gain, we select a seed node with marginal
contribution that is within a small relative error $\epsilon$ of the
maximum with high probability.  It also suffices to require 
that the relative error is bounded by $\epsilon$ in 
expectation and is concentrated, that is, $\forall a>1$, 
the probability of error exceeding
$a\epsilon$ decreases exponentially in $a$.
It turns out that the approximation
ratio of {\sc Approximate Greedy} is $1-(1-1/s)^s -O(\epsilon)$ with
corresponding guarantees~\cite{binaryinfluence:CIKM2014}.

  The \skim\ algorithm applies approximate greedy to
  reachability-based influence.  It works  with partial sketches (samples) from
  ``influence sets'' of nodes to determine a node with approximately maximum
  marginal gain.  These partial sketches need to be updated very
  efficiently after each seed is selected. 
  A critical component for the scalability and accuracy of \skim\ is
  sampling with respect to the residual problem.   This is because the
  marginal influence of a node can be much smaller than its initial
  influence and we can get good accuracy with a small sample only if
  we use the residual problem.

 The residual problem (and current samples) are updated
 after each seed selection, both with exact and approximate {\sc
   Greedy}.   We now consider the total number of edge traversals used
 in these updates.
With reachability-based influence,  with exact or approximate {\sc Greedy}, the
number of traversals is  linear in input size: This is because there can be at most one search which
progresses through a node in each instance. Once a node is reachable
from the seed set in that instance, it is influenced, and everything
reachable from the node in the same instance is reachable, and influenced, as well.  So
these nodes never need to be visited again and can be removed.

  This is not true, however, for distance-based influence:
For each node-instance pair $(v,i)$, 
the distance $\delta^{(i)}_{v}$ can be updated many times.
Moreover, when the distance to the seed set
 decreases, as a result of adding a new seed node, the node $v$ and its outgoing edges in the 
 instance $i$ are traversed.
  Therefore, to bound the number of edge traversals we must 
bound the number of updates of $\delta^{(i)}_{v}$.

For exact {\sc Greedy}, there are  pathological
inputs with $\Omega(sn)$ updates, which roughly translates to
$\Omega(s|E|)$ edge traversals, even on a single instance and
threshold influence.
\notinproc{\footnote{Consider a deterministic instance (directed
    graph) and threshold influence with $T>1$.  The nodes $\{a_i\}$ have edges to all nodes in a set
      $M$ with length $T(1-i/(2|A|))$ and special nodes
      $\{b_1,\ldots,b_i\}$ with distance $1$.  The greedy selection,
      and in fact also the optimal one,
      selects the nodes $\{a_i\}$ by decreasing index.  But each
    selection updates the $\delta$ for all nodes in $M$. We can choose
  number of $a$ nodes roughly $\sqrt{|M|}$ and number of $b$ nodes
  roughly $|M|$.  Obtaining that each of the first $\sqrt{n}$ seeds
  (where $n$ is the total number of nodes) results in $\delta$ updates
for $\Omega(n)$ nodes.} On realistic inputs, however, we expect the number of
  updates to be no more than $\ln s$ per node.
This is because even if seed nodes are added in random order, the
expected number of times the distance to the seed set decreases is
well concentrated around $\ln s$.   Therefore, since selected seeds
should typically be closer than other nodes, we would expect a small
number of updates.}   
  Remarkably, we show that our {\sc Approximate Greedy} selection
circumvents this worst-case  behavior of the
exact algorithm and guarantees a near-linear number of updates \onlyinproc{ (proof provided in the full version)}:
\begin{theorem} \label{randupdatebound:thm}
Suppose that the approximate greedy selection has the
  property  that for some $\epsilon$, the next seed is selected in a near
  uniform way from all nodes with marginal influence that is at least
  $(1-\epsilon)$ of the maximum.  Then the expected total
  number of updates of $\delta^{(i)}_v$ at a node-instance pair
  $(v,i)$ is bounded by $O(\epsilon^{-1}\log^2 n)$.
\end{theorem}
\notinproc{
\begin{proof}
 Consider the candidate next seeds $S'$ (those with marginal influence
 that is within
$(1-\epsilon)$ of the maximum).  For a node
instance pair $(v,i)$,
consider all nodes $u\in S'$ with $d_{uv}< \delta^{(i)}_v$.  Then the
probability that the $j$th selected node from this group would be
closer than all previously selected ones is $1/j$.  We obtain that the
expected number of updates to $\delta^{(i)}_v$ from nodes in this
group is at most $\ln n$.  We now consider, for the sake of analysis
only, partitioning the greedy selection to steps, according to a
decrease by a factor of $(1-\epsilon)$ of the
maximum marginal influence of a node.  We showed that the expected
number of selections in a step is at most $\ln n$.  Also note that
the marginal influence of a node can only
decrease after another node is selected, so it is possible that it moves to the next step.
We now bound the number of steps.  Because we are interested in
$\epsilon$ approximation, we can stop the greedy algorithm once
the maximum marginal influence is below $\epsilon/n$ of the initial
one.  Therefore, there are at most $\epsilon^{-1}\log n$ steps.
\end{proof}
}

\section{Threshold influence} \label{thresh:sec}

We present  both an oracle and an approximate greedy IM algorithm  for
threshold influence~\cite{Gomez-RodriguezBS:ICML2011,DSGZ:nips2013}.  In this model, 
a node is considered influenced by $S$ in an instance if it is
  within distance at most $T$ from the seed set. Formally, $\alpha(x)=1$
  when $x\leq T$ and $\alpha(x)=0$ otherwise.   
\notinproc{The simpler structure
  of the kernel allows for simpler algorithms, 
  intuitively,  because the contributions to influence
  $\Delta^{(i)}_{uv}$ are in $\{0,1\}$, as with reachability-based influence.}

\subsection{Threshold-Influence Oracle} 

 Our
influence oracle for a prespecified
  threshold $T$ generalizes the
reachability-based influence oracle of Cohen et al.~\cite{binaryinfluence:CIKM2014}.
The reachability-based influence oracle preprocesses the input to compute a {\em combined reachability sketch} for
each node.  Each node-instance pair is assigned a random permutation
rank (a number in $[n \ell]$) and the combined reachability sketch of
a node $u$ is a set consisting of the $k$ smallest ranks amongst node-instance pairs
$(v,i)$ such that $v$ is reachable from $u$ in instance $i$.
This is also called a bottom-$k$ sketch of reachable pairs. 
The oracle uses the sketches of the nodes in $S$ to estimate their
influence by applying 
the union size estimator
for bottom-$k$ sketches~\cite{CK:sigmetrics09}. 
The combined reachability sketches are
built
by first computing a set of reachability sketches (one for each node)~\cite{ECohen6f} in each instance
and then combining, for each node,  the sketches obtained in different
instances to obtain one size-$k$ sketch.   In turn, the computation
for each instance 
uses
reverse (backward) 
reachability computations.  The algorithm of Cohen~\cite{ECohen6f}
initiates these reversed reachability searches  from all nodes in a
random permutation order.  These searches  are pruned at nodes
already visited $k$ times.

For threshold influence, we  instead
consider a pair $(v,i)$ reachable from $u$ if $d^{(i)}_{uv}\leq T$.
We then compute for each node  the bottom-$k$ sketch of these
``reachable'' pairs under the modified definition.  The oracle
estimator~\cite{CK:sigmetrics09} is the same one used for the reachability-based
case; the estimate has (worst-case) CV that is at most 
$1/\sqrt{k-2}$ with good concentration.
The computation of the sketches is nearly as efficient as for the
reachability-based case.
 Instead of using reverse reachability searches, 
for threshold influence we use
{\em reverse Dijkstra} computations (single-source shortest-path searches on the graph with reversed edges).
These computations are pruned both at distance $T$ and (as with
reachability sketches) at nodes already visited $k$ times.
The sets of sketches obtained for the different instances are
combined as in~\cite{binaryinfluence:CIKM2014} to obtain a set of combined sketches
(one combined sketch with $k$ entries for each node).

The running time is dominated by the computation of the sketches.
The preprocessing computation is $O(k \sum_{i=1}^\ell  |E^{(i)}| \log n)$, the sketch representation is $O(kn)$, 
and each influence query for a set $S$ takes
$O(|S|k \log |S|)$ time.  

\onlyinproc{As in~\cite{binaryinfluence:CIKM2014} for
  reachability-based influence in the IC model, we can construct an
  oracle of the same size and approximate guaranees for a distance-based IC model.
 The algorithm, however, uses $kn$ reverse searches (see full version for details).}
 \notinproc{

\paragraph{Oracle for a distance-based IC model} 
A distance-based Independent Cascade (IC) model $\mathcal{G}$ is specified by associating 
an independent  random length $w(e) \in [0,+\infty]$
with each edge $e$ according to a distribution that is associated with 
the edge. 
The probability that $e$ is live is $p_{e}=\Pr[w(e)<\infty]$
and its length if it is live is $w(e)$. We use the convention that an 
edge $e$ that is not explicitly specified has $w(e)\equiv \infty$ (is never live). 

 As with the IC model for reachability-based influence \cite{binaryinfluence:CIKM2014,LOS:kdd2015},
  we can compute oracles based on sketches that have the same size,
  query time, and approximation guarantees, but with respect to a distance-based IC
  model (where the exact influence is defined as the expectation over
  the model instead of the
  average of $\ell$ deterministic instances).   The worst-case computation
  involved, however,  is higher.
The algorithm performs up to $k=O(\epsilon^{-2}\log n)$ 
reverse Dijkstra searches
  (pruned to distance $T$) from each node.  Each search randomly instantiates
incoming edges of scanned nodes according to the model.  Pseudocode is
  provided as Algorithm~\ref{ICoracle:alg}.
Note that in contrast to 
sketching instances, the searches are not pruned at
nodes that already have a size-$k$ sketch.   Therefore we can not
bound the total number of node scans by as we did with instances, and
have a worst-case bound of $O(kn^2)$.    On the positive size, we do
maintain the property that the sketch sizes are $k$.
 The correctness  
  arguments of the approximation guarantees with respect to the 
  expectation of the IC model carry over from \cite{binaryinfluence:CIKM2014}. 

\begin{algorithm2e}[h]\caption{Build sketches $S(v)$, $v\in V$ for an IC model\label{ICoracle:alg}}
\DontPrintSemicolon
\lForEach(\tcp*[f]{Initialize sketches}){$v\in V$}{$S(v)\gets \emptyset$}
\tcp{Iterate over $k$ random permutations of the nodes $V$}
\For{$j = 0,\ldots,k-1$}
{
Draw a random permutation of the nodes $\pi: [n] \rightarrow V$\;
\For{$i=1,\ldots,n$}
{
Perform a reverse Dijkstra search from node $\pi(i)$ pruned to distance
$T$\;
\ForEach{scanned node $u$}
{
\If{ $|S(u)| < k$}
{$S(u) \gets S(u) \cup \{n*j+i\}$ \;
\If{all nodes $z$ have $|S(z)|=k$ }{exit}
Instantiate and process length and presence of incoming edges of $u$\; 
}
}
}
}
\Return{$S(v)$ for all $v\in V$}
\end{algorithm2e}

 }


\subsection{Threshold-Influence Maximization} 
Our algorithm for threshold influence maximization, which we call \tskim, generalizes
\skim~\cite{binaryinfluence:CIKM2014}, which was designed for 
reachability-based influence.  \notinproc{A pseudocode for \tskim\ is provided as 
Algorithm~\ref{Tskim:alg} in Appendix~\ref{app:ps:alg}.}

Our algorithm \tskim\ 
builds sketches,
but only to the point of determining the node
with maximum estimated influence.  We then compute a residual problem
which updates the sketches.  \tskim\ build sketches using
reverse single-source shortest path computations that are pruned at 
distance $T$ (depth-$T$ Dijkstra).
 As with exact greedy for distance-based
influence (Section \ref{model:sec}), \tskim\ maintains a residual
problem.  This requires updating the distances
$\Covered[v,i]=d^{(i)}_{Sv}$ from the current seed set $S$, as in $\AddSeed(u)$,
and also updating the sketches to remove the contributions of pairs
that are already covered by the seed set.

  The (worst-case) estimation quality guarantee of \tskim\ is
similar to that of \skim. When using $k=O(\epsilon^{-2} \log n)$
we obtain that, with high probability (greater than $1-1/\text{poly}(n)$), for
all $s\geq 1$, the
influence of the first $s$ selected nodes is at least $1-(1-1/s)^s-\epsilon$ of the maximum influence of
a seed set of size $s$.  \onlyinproc{The computation timeis
  near-linear and analysis is
  provided in the full version.}\notinproc{The computation time analysis of \tskim\ is deferred to
Appendix \ref{tskimtime:sec}.}

\section{Influence oracle} \label{timedQ:sec}

We now present our oracle for distance-based influence, as defined in Equation
\eqref{timedinf}.
We preprocess the input $\mathcal{G}$ to 
compute a sketch $X_v$ for each node $v$.  Influence queries,
which are specified by a seed set $S$ of nodes and {\em any} function
$\alpha$,  can be
approximated from the sketches of the query seed nodes.

\notinproc{
 Note that the same  set of sketches can be used 
 to estimate distance-based influence with respect to any non-increasing function 
 $\alpha$.  That is, $\alpha$ can be specified on the fly, after the 
sketches are computed.  
When we are only interested in a specific $\alpha$, such as a
threshold function with a given $T$ (Section \ref{thresh:sec}) or
reachability-based influence \cite{binaryinfluence:CIKM2014}, the sketch size  and 
construction time can be reduced.

 In the following we present in detail the three components of our oracle: 
the definition of the sketches, estimation of 
 influence from sketches (queries), and building the sketches
 (preprocessing).

 The sketches are defined in Section \ref{cADS:sec} and we show that
for an  input specified as either a set of instances or as a distance-based IC model,
each sketch $X_v$ has a (well concentrated) expected size that is at
most $k\ln (nk)$.   The total storage of our oracle is therefore $O(nk\log(nk))$.

 The sketch-based influence estimator is presented in Section \ref{infest:sec}.
We establish the following worst-case bounds on estimation
quality.  
\begin{theorem} \label{timedinforacles:thm}
Influence queries $\INF(\mathcal{G},S)$, specified by a 
set $S$ of seed nodes and a function $\alpha$, can be estimated in $O(|S|k\log n)$
time from the 
sketches $\{X_u \mid u\in S\}$.  The 
estimate is nonnegative and unbiased, has CV $\leq 1/\sqrt{2k-2}$,  and is 
well concentrated (the probability that the relative 
error exceeds $a/\sqrt{k}$ decreases exponentially with $a>1$). 
\end{theorem}
We also show that our estimators  are designed to fully exploit the information in
the sketches in an instance-optimal manner.  
\notinproc{
 The preprocessing is discussed in Section \ref{combads:sec}.  We
 show that
 for a set of $\ell$ instances $\mathcal{G}=\{\cascade{i}\}$, the 
 preprocessing is performed in expected time $O(k \sum_{i=1}^\ell 
 |E^{(i)}| \log n)$.  }

\ignore{
Theorem \ref{timedinforacles:thm}:
\begin{proof}
The 
bounds for a single seed follow from the analysis of the HIP estimator
in \cite{ECohenADS:PODS2014}.  The worst-case bound for a larger seed set 
follows from \cite{ECohenADS:PODS2014} in combination with the 
optimality of the L$^*$ estimator \cite{sorder:PODC2014}. 
\end{proof}
}
} 

\onlyinproc{
We present here the definition
of the sketches. 
For an  input specified as either a set of instances or as a distance-based IC model,
each sketch $X_v$ has a (well concentrated) expected size that is at
most $k\ln (nk)$ 
(and $k\ln(n\min\{k,\ell\})$ when using $\ell$ instances).
 The total storage of our oracle is therefore $O(nk\log(nk))$
(and $nk\ln(n\min\{k,\ell\})$ when using $\ell$ instances).

In the full version we detail our influence estimator, which optimally uses the information in the
sketch, and the algorithm to compute the sketches (preprocessing).   We show 
that for a set of $\ell$ instances $\mathcal{G}=\{\cascade{i}\}$, the 
expected time is $O(k \sum_{i=1}^\ell 
 |E^{(i)}| \log n)$.  We establish
  the following worst-case bounds on estimate quality:
\begin{theorem} \label{timedinforacles:thm}
Influence queries $\INF(\mathcal{G},S)$, specified by a 
set $S$ of seed nodes and a function $\alpha$, can be estimated in $O(|S|k\log n)$
time from the 
sketches $\{X_u \mid u\in S\}$.  The 
estimate is nonnegative and unbiased, has CV $\leq 1/\sqrt{2k-2}$,  and is 
well concentrated (the probability that the relative 
error exceeds $a/\sqrt{k}$ decreases exponentially with $a>1$). 
\end{theorem}
}
\notinproc{

\subsection{Combined ADS} \label{cADS:sec}
}

 Our  {\em combined} All-Distances Sketches~($\cADS$) are a
 multi-instance generalization of All-Distances Sketches ($\ADS$)
 \cite{ECohen6f,bottomk07:ds,ECohenADS:PODS2014} and build on the
 related combined reachability sketches
 \cite{binaryinfluence:CIKM2014} used for reachability-based influence.

The $\cADS$ sketches are randomized structures defined with respect to 
random rank values $r^{(i)}_u 
  \sim U[0,1]$ associated with each node-instance pair $(u,i)$. 
To improve estimation quality in practice, we
restrict ourselves to a particular form of
{\em
  structured permutation ranks}
\cite{binaryinfluence:CIKM2014}:
For a set of instances, the ranks are a permutation of $1,\ldots,n 
\min\{\ell,k\}$, where each block of positions of the form
$in,(i+1)n-1$ (for integral $i$) corresponds
to an independent  random permutation of the nodes. For each node $u$, the instances
$i_j$ in $r^{(i_j)}_u$, when ordered  by
increasing rank, are a uniform random selection (without replacement).

For each node $v$,  $\cADS(v)$ is a set of 
rank-distance pairs of the form $(r^{(i)}_{u},d^{(i)}_{vu})$ which 
includes $\min\{\ell,k\}$ pairs of distance $0$, that is,
all such pairs if $\ell\leq k$ and the $k$ smallest rank values 
otherwise. 
It also includes pairs with positive distance when the rank value is at most the $k$th 
smallest amongst closer nodes (across all instances).
Formally,
\begin{equation} \label{cADS:def}
\cADS(v)=\left\{ \begin{array}{l} 
 \bigl\{ (r^{(i)}_v,0) \mid r^{(i)}_v\in \text{{\sc 
  bottom-}}k\{ r^{(j)}_v | j\in [\ell] \} \bigr\} 
\\
\bigl\{ (r^{(i)}_{u},d^{(i)}_{vu}) \mid r^{(i)}_u < \kth_{(y,j)\mid d^{(j)}_{vy}< d^{(i)}_{vu}} r^{(j)}_y \bigr\} \end{array} \right.. 
\end{equation}
Here {\sc bottom}-$k$ refers to the smallest $k$ elements in the set and
$\kth$ denotes the $k$th smallest element in the set.
 When
there are fewer than $k$ elements, we define $\kth$ as the
domain maximum.
For the purpose of sketch definition, we treat all positive 
distances across instances as unique\notinproc{; we apply some arbitrary tie-breaking,
for example according to node-instance index, when this is not the case}.

\notinproc{
The size of $\cADS$ sketches is a random variable,  but we can 
bound its expectation.
Moreover, $|\cADS(u)|$ is well 
concentrated  around the expectation.  The proof resembles 
that for the basic ADS~\cite{ECohen6f}.  
\begin{lemma}
 $\forall u,\,  \E[|\cADS(u)|] \leq k \ln(n\min\{k,\ell\})$. 
\end{lemma}
\begin{proof}
We consider all 
node-instance pairs $(v,i)$ such that 
$r^{(i)}_v \in \text{{\sc bottom-}}k\{ r^{(j)}_v | j\in [\ell] \}$
by increasing distance from $u$. 
The probability that the $j$th item contributes to $\cADS(u)$ is 
the probability that its rank is amongst the $k$ smallest in the first 
$j$ nodes, which is $\min\{1,k/j\}$.  Summing over $i\leq |R_u|$ we obtain the bound. 
\end{proof}
} 

Note that  we can also define $\cADS$ sketches with respect to a
probabilistic model. The definition
  emulates working with an infinite set of instances generated
  according to the model.  Since there are at most $nk$ distinct rank
  values in the sketches, and they are all from the first $nk$
  structured permutation ranks, the entries in the sketches are
  integers in $[nk]$.

\notinproc{
\subsection{Estimating Influence} \label{infest:sec}

  Our estimators use 
the {\em HIP threshold}, $\tau_v(x)$, defined with respect to a node $v$ and a positive distance value $x> 0$~\cite{ECohenADS:PODS2014}:
\begin{eqnarray}
\tau_v(x) &=& \kth_{(u,j) \mid d^{(j)}_{vu} < x} r^{(j)}_{u}
\nonumber \\ 
&=& \kth\{ r \mid (r,d)\in \cADS(v) \text{ and } 
d < x \}.\label{threshrelADS:eq}
\end{eqnarray}
The value $\tau_v(x)$ is the $k$th smallest  rank value amongst 
  pairs $(y,j)$ whose distance $d^{(j)}_{vy}$ is smaller than $x$.
If there are fewer than $k$
  pairs with distance smaller than $x$, then the
  threshold is defined to be the maximum value in the rank domain.
Note that, since the $\cADS$ contains the $k$ smallest ranks within 
 each distance, $\tau_v(x)$ is  also 
  the $k$th smallest amongst such pairs that are in $\cADS(v)$.
  Therefore, the threshold values $\tau_v(x)$
can be  computed  from $\cADS(v)$ for all $x$.  

\notinproc{
  The HIP threshold has the following interpretation.
For a node-instance pair $(u,i)$, 
$\tau_v(d^{(i)}_{vu})$ is
the largest rank value $r^{(i)}_u$ that would allow the (rank of the)  pair 
to be included in $\cADS(v)$, conditioned on fixing the ranks of all 
other pairs.  
We now can consider the probability that the 
pair $(y,j)$ is included in $\cADS(v)$, fixing the ranks of all 
pairs other than $(y,j)$. This is exactly the probability that a 
random rank value is smaller than the HIP threshold $\tau_v(x)$.
In particular, if $\tau_v(x)$ is the domain maximum, the inclusion
probably is $1$.    We refer to this probability as the HIP inclusion probability.
}

When ranks are uniformly drawn from $[0,1]$ the
HIP inclusion probability is equal to the HIP threshold $\tau_v(x)$ and we use
the same notation.
When we work with integral structured ranks, we 
divide them by $n\ell$ to obtain values in $[0,1]$.

We now present the estimator for the influence 
\begin{eqnarray} 
\lefteqn{\INF(\{\cascade{i}\},S)=\frac{1}{\ell} \sum_{(v,i)}
  \max_{u\in S} \alpha(d^{(i)}_{uv})} \nonumber\\
&=&  \frac{1}{\ell} \sum_{(v,i)} \alpha(\min_{u\in S}d^{(i)}_{uv})=
\frac{1}{\ell} \sum_{(v,i)} \alpha(d^{(i)}_{Sv}) \label{relinf:eq}
\end{eqnarray}
from $\{\cADS(u) \mid u\in S\}.$
We first discuss $|S|=1$.  We use the HIP estimator
\cite{ECohenADS:PODS2014}
\begin{equation} \label{singleseedest}
\widehat{\INF}(\{\cascade{i}\},u)=\alpha(0)+\frac{1}{\ell}\sum_{(r,d)\in \cADS(u) \mid d>0}
\frac{\alpha(d)}{\tau_u(d)},
\end{equation}
which is the sum over ``sampled'' pairs $(r,d)$ (those 
included in $\cADS(u)$) of the contribution $\alpha(d)/\ell$ of the pair to
the influence of $u$, divided by the HIP inclusion probability of the
pair.\onlyinproc{ We can show that this estimator has CV at most
  $1/\sqrt{2k-2}$ (see full version)}\notinproc{
\begin{lemma}
The estimator \eqref{singleseedest} is unbiased and has
CV that is at most $1/\sqrt{2k-2}$. 
\end{lemma}
\begin{proof}
A node always influences itself (the only node of distance $0$ from
it), and the estimate for that contribution is $\alpha(0)$.  We apply the HIP estimator of \cite{ECohenADS:PODS2014} to estimate the 
contribution of nodes with positive distance from $u$.  For a pair 
$(v,i)$, the HIP 
estimate is $0$ for pairs not in $\cADS(u)$.  When the pair is in
$\cADS(u)$, we can compute the HIP probability $\tau_u(d^{(i)}_{uv})$
and obtain the estimate
$\alpha(d^{(i)}_{uv})/\tau_u(d^{(i)}_{uv})$.  Since we are considering
node-instance pairs, we divide by the number of instances $\ell$.
The variance analysis is very similar to \cite{ECohenADS:PODS2014}.
\end{proof}
} 


We now consider a seed set $S$ with multiple nodes. 
The simplest way to handle such a set is to
generate the \emph{union $\cADS$},
which is a $\cADS$ computed with respect to the minimum 
distances from any node in $S$.  The union $\cADS$ can be computed by
merging the $\cADS$ of the seed nodes \onlyinproc{and then applying
  the single-seed estimator to that. }\notinproc{using a similar procedure to
Algorithm~\ref{ADScomb:alg} (in Appendix~\ref{timedoracle:ps:alg}) which will be presented in
Section \ref{combads:sec}.
We then estimate the 
contribution of the nodes in $S$ by $|S|\alpha(0)$ and estimate the contribution of all nodes that have a positive distance from $S$ by 
applying the HIP estimator to the entries in the union $\cADS$.   This
estimator has the worst-case bounds on estimation
quality claimed in Theorem \ref{timedinforacles:thm}, but discards a
lot of information in the union of the sketches which could be used to
tighten the estimate.

\subsubsection{Optimal Oracle Estimator}}
The estimator we propose and implement \onlyinproc{instead }uses the information in 
 the sketches of nodes in $S$ in an optimal way. This means the
 variance can be smaller, up to a factor of $|S|$, than that of the
 union estimator.\onlyinproc{ Details are given in the full version.}\notinproc{  A pseudocode is 
 provided as Algorithm~\ref{TIoracle:alg} in Appendix~\ref{timedoracle:ps:alg}. 
The estimator first computes the set $Z$ of rank values $r$ that
appear with distance $0$ in at least one sketch.
 These ranks correspond to node-instance pairs involving a seed node.
For each rank value $r$ that appears in at least one sketch in $S$ and is
not in $Z$ (has positive distance in all sketches), we build the set $T_r$ of 
threshold-contribution pairs that correspond to occurrences of $r$ in sketches of $S$. 
We then compute from $T_r$
the sorted {\em skyline} (Pareto set)
$\Skylines{r}$ of $T_r$.   The skyline $\Skylines{r}$ includes a pair $(\tau,\alpha)\in T_r$ 
if and only if the pair is not {\em dominated} by any other pair. 
That is, any pair with a larger $\tau$ value must have a smaller 
$\alpha$ value.
We compute $\Skylines{r}$ from
$T_r$ as follows. We first sort $T_r$ lexicographically, first by decreasing
$\tau$, and then if there are multiple entries with same $\tau$, by
decreasing $\alpha$. We then obtain $\Skylines{r}$ by a linear scan of the sorted $T_r$ which
removes pairs with a lower $\alpha$ value than the maximum $\alpha$
value seen so far.
The entries of the computed $\Skylines{r}$ are sorted by 
decreasing $\tau_j$ (and increasing $\alpha_j$).  

  For each $r$ for which we computed a skyline (appears in at least one
  sketch of a node in $S$ and is not in $Z$), we apply the 
\lstar\ estimator
 \cite{sorder:PODC2014}  
 to the sorted $\Skylines(r) = \{(\tau_j,\alpha_j)\}$. 
The pseudocode for \lstar\ tailored to our
application is in Algorithm~\ref{Lest:alg} in
Section~\ref{timedoracle:ps:alg}, and details
on the derivation and applicability of the estimator 
are provided in Appendix \ref{oracleLstar:sec}.

Finally, the influence
estimate $\widehat{\INF}(\{\cascade{i}\},S)$ returned by
Algorithm~\ref{TIoracle:alg} has two components. 
The first summand ($|S|\alpha(0)$) is the contribution of the seed nodes
themselves. 
The second component is the sum, over all node-instance pairs of
positive distance from $S$, of their estimated contribution to the
influence (normalized by the number of instances $\ell$).  We estimate
this by the sum of the \lstar\ estimates applied to $\Skylines(r)$.
} 

}
\notinproc{
\subsection{Computing the Set of cADS} \label{combads:sec}

 We compute a set of $\cADS$ sketches by computing a set of $\ADS$
 sketches $\ADS^{(i)}(v)$ for each instance
 $i$~\cite{ECohen6f,ECohenADS:PODS2014}.
The computation of sketches for all nodes $v$ in a single instance 
 can be done using
{\sc Pruned Dijkstra}s~\cite{ECohen6f,bottomk07:ds} or the node-centric {\sc Local Updates}~\cite{ECohenADS:PODS2014}.

An $\ADS$ is a $\cADS$ of a single instance and has the same basic
form: a list of rank-distance pairs sorted by increasing distance. 
It can have, however, at most one entry of distance $0$.  

For each node $v$, we compute $\cADS(v)$ by combining $\ADS^{(i)}(v)$
for all instances $i$.\notinproc{A pseudocode for combining two rank-distance lists to a $\cADS$
format list is provided as 
Algorithm \ref{ADScomb:alg}.} The algorithm can be applied repeatedly
to $\ADS^{(i)}(v)$ and the current $\cADS(v)$, or in any combination
order of rank-distance lists to obtain  the same end result.

The computation of the set of $\cADS$ sketches is dominated by computing a set of 
All-Distances Sketches \cite{ECohen6f,bottomk07:ds,ECohenADS:PODS2014}
in each instance.  The computation for instance $i$ takes $O(k
|E^{(i)}| \log n)$ time.

\ignore{
\begin{algorithm2e}[h]\caption{$\cADS$ Merge (OLD)}
\DontPrintSemicolon
\SetKwFunction{Merge}{merge}
 \Merge{$A_1,A_2$} \tcp{Combine two sorted rank distance lists $A_1$ and $A_2$ into a 
 single $\cADS$ format list}\BlankLine
$R_0 \gets \{x \mid  (x,0)\in A_1 \cup A_2\}$\;
$L \gets$ merge positive distance entries in $A_1 
\cup A_2$ by increasing distance \;
$Q\gets \perp$ \tcp{Initialize MAX priority queue of 
  size $k$}\;
$A \gets \{(r,0) \mid r\in \text{{\sc bottom-}}k \{R_0\} \}$\;
\For {$r \in R_0$}{insert $r$  into priority queue $Q$\;}
\While {$L\not=\emptyset$}
{
$(r,d)\gets L.next$\;
\If {$r< L.\max$ or $|L|<k$} 
{
Insert $(r,d)$ to $A$\;
Insert $r$ to $L$\;
\If{$|L|>k$}{Remove max element from $L$}
}
}
\Return{$A$}\;
\end{algorithm2e}
}

} 
\ignore{
\subsection{Sketches for a distance-based IC model}

 We build $\cADS$ sketches for a distance-based IC model as follows.
We first assign $k$ structured permutation ranks from $[nk]$ to each node.
At each step we choose the next node in rank order and perform 
Dijkstra on a random transpose instance generated according to the
model.  
The $\cADS$ of visited 
 nodes are augmented when the rank is one of the $k$ smallest amongst 
 nodes that are at least as close.  However, we need to propagate the
 computation regardless until all nodes have sufficient $\cADS$
 sketches.
}

\begin{algorithm2e}
\caption{Distance-based IM  (\alphaskim) \label{alphaskim:alg}}
{\scriptsize 
\DontPrintSemicolon 
\KwIn{Directed graphs $\{G^{(i)}(V,E^{(i)},w^{(i)})\}$, $\alpha$}
\KwOut{Sequence of node and marginal influence pairs}
\SetKwArray{Qgrand}{Qpairs} 
\SetKwArray{Qcands}{Qcands} 
\SetKwArray{Qhml}{Qhml} 
\SetKwArray{Index}{index}
\SetKwArray{SeedList}{seedlist}
\SetKwArray{Rank}{rank}
\SetKwArray{cdelta}{$\delta$}
\SetKwArray{HM}{HM}
\SetKwArray{ML}{ML}
\SetKwFunction{NextSeed}{NextSeed} 
\SetKwFunction{MoveUp}{MoveUp} 
\SetKwFunction{MoveDown}{MoveDown} 
\SetKwFunction{UpdateReclassThresh}{UpdateReclassThresh} 
\SetKwArray{est}{Est}
\tcp{Initialization}
$\Rank \gets $ map node-instance pairs 
$(v,i)$ to $j/(n\ell)$ where $j\in [n\ell]$;\\
\lForAll(\tcp*[f]{List of  $(u,d^{(i)}_{uv})$ scanned by reverse Dijkstra $(v,i)$}){node-instance pairs $(v,i)$}{$\Index{v,i} \leftarrow \perp$}
\lForAll(\tcp*[f]{Distance from $S$}){pairs $(u,i)$}{$\cdelta{u,i} \gets  \infty$} 
\lForAll{nodes~$v$}{$\est.H[v] \leftarrow 0$; $\est.M[v]\gets 0$}
$\SeedList \leftarrow \perp$\tcp*{List of seeds \& marg.\ influences} 
\lForAll(\tcp*[f]{Initialize $\Qgrand$}){node-instance pairs $(v,i)$}{Insert 
$(v,i)$ to $\Qgrand$ with priority 
$\alpha(0)/\Rank{v,i}$}
$s\gets 0$;  $\tau  \gets \alpha(0) n\ell/(2k)$;
$coverage \gets 0$ \tcp*{coverage of current seed set} 
\While{$coverage < n\ell \alpha(0)$}
{
\tcp{Build PPS samples of marginal influence sets until confidence in 
  next seed}
\While{$((x,\hat{I}_x) \gets \NextSeed()) =\perp$}
{
$\tau \gets \tau \lambda$;
$\MoveUp{}$ \tcp*{Update est. components} 
\ForAll{pairs $(v,i)$  in $\Qgrand$ with priority $\geq \tau$}
{Remove $(v,i)$ from $\Qgrand$ \;
 Resume reverse Dijkstra from $(v,i)$.  
\ForEach{scanned node $u$ of distance $d$}
{
$c \gets \alpha(d)-\alpha(\cdelta{v,i})$;\\
\lIf{$c\leq 0$}{Terminate rev. Dijkstra from $(v,i)$}
\If{$c/\Rank{v,i} < \tau$}{
  place $(v,i)$ with priority $c/\Rank{v,i}$ in $\Qgrand$;
 Pause reverse Dijkstra from $(v,i)$
}
\Else(\tcp*[f]{$c/\Rank{v,i}\geq \tau$}){
\Append $(u,d)$ to $\Index{v,i}$ \\
\lIf(\tcp*[f]{H entry}){$c \geq 
  \tau$}{$\increase{\est.H[u]}{c}$}\lElse{\tcp*[f]{M 
    entry} \\ $\increase{\est.M[u]}{1}$
 \\ \If{$\HM{v,i}=\perp$}{$\HM{v,i}\gets |\Index{v,i}|$;  Insert $(v,i)$ with priority $c$ to $\Qhml$}
}
Update the priority of $u$ in $\Qcands$ to $\est.H[u]+\tau \est.M[u]$
}
}
}}
\tcp{Process new seed node $x$}
$I_x \gets 0$ \tcp*{Exact marginal influence} 
\ForEach{instance $i$}{Perform a forward Dijkstra from $x$ in 
  $G^{(i)}$.  \ForEach{visited node $v$ at distance $d$}{\lIf{$d \geq 
      \cdelta{v,i}$}{Prune}\Else{
$\decrease{\text{priority} (v,i) \text{ in 
  }\Qgrand}{\frac{\alpha(d)-\alpha(\cdelta{v,i})}{\Rank{v,i}}}$;\\
\lIf{priority of $(v,i)$ in $\Qgrand$ $\leq0$}
{terminate rev. Dijkstra $(v,i)$ and remove $(v,i)$
      from $\Qgrand$}
       \MoveDown{$(v,i), \cdelta{v,i}, d$}; \\ $\increase{I_x}{\alpha(d)-\alpha(\cdelta{v,i})}$; $\cdelta{v,i} \gets d$ }
}}
$\increase{s}{1}$; $\increase{coverage}{I_x}$; \SeedList.\Append{$x$,$\hat{I}_x/\ell$,$I_x/\ell$}
}
\Return{\SeedList}
} 
\end{algorithm2e}

\section{Influence Maximization} \label{timedIM:sec}

We present \alphaskim\ (pseudocode in Algorithm \ref{alphaskim:alg}), 
which computes an approximate greedy sequence for distance-based influence 
with respect to an arbitrary non-increasing $\alpha$. 
The input to \alphaskim\ is a set of instances $\mathcal{G}$ and a 
decay function $\alpha$. 
The output is a sequence of nodes, so that each prefix approximates 
the maximum influence seed set of the same size.

Like exact greedy (Section~\ref{model:sec}) and \tskim\
(Section~\ref{thresh:sec}), \alphaskim\ maintains a residual problem, specified by
the original input $\mathcal{G}$ and distances $\delta^{(i)}_v$.  It
also maintains, for each node, a sample of its influence
set, weighted by the respective contribution of each element.  
The sampling is governed by a global sampling threshold $\tau$,
which inversely determines the inclusion probability in the sample
(the lower $\tau$ is, the larger is the sample).
The weighted sample has the same role as the partial sketches
maintained in \tskim, as it allow us to estimate the influence of nodes.

At a high level,  \alphaskim\  alternates between two subroutines.
The first subroutine examines the influence estimates of nodes.  
We pause if we have sufficient confidence that the node with the maximum
{\em estimated} influence (in the current residual problem)
 has actual influence that is sufficiently
close to the maximum influence.  Otherwise, we decrease 
$\tau$, by multiplying it by a (fixed)
$\lambda < 1$ (we used $\lambda=0.5$ in our implementation),  extend the samples,  and update the estimates on
the influence of nodes to be with respect to the new threshold $\tau$.
We pause only when we are happy with the node with maximum estimated influence.

The second subroutine is invoked when a new node $x$ is selected to be added to
the seed set;
\alphaskim\ updates the residual problem, that is, the distances
$\delta$ and the samples.

\onlyinproc{See the full version for more details on \alphaskim\ and its analysis.
\begin{theorem} 
using $k=O(\epsilon^{-2} \log n)$, \alphaskim\ runs in expected time 
$$O\left(\frac{\log^3 n}{\epsilon}\sum_{i=1}^\ell |E^{(i)}| + \frac{\log^3
  n}{\epsilon^3}n+ \frac{\log n}{\epsilon^2}|\bigcup_{i=1}^\ell
E^{(i)}|\right)$$ and returns a
sequence of seeds so that for each prefix $S$ of size $s$, with high
probability,
\begin{equation} \label{approxguarantee}
\INF(\mathcal{G},S) \geq (1-(1-1/s)^s -\epsilon) \max_{ U \mid |U|= s}
\INF(\mathcal{G},U).
\end{equation}
\end{theorem}
}
We provide an overview of our presentation of  the components of \alphaskim.
In Section \ref{pps:sec}  we precisely
define the weighted samples we use.
 In Section \ref{index:sec} we
present our main data structure, $\Index$,  which stores the (inverted)
samples.
The building of $\Index$, which dominates the computation, is done
using applications of pruned reverse Dijkstra, discussed in Section 
\ref{rdijkstra:sec}. \notinproc{The selection of the next seed node 
is detailed in Section~\ref{nextseed:sec}.}
The samples are defined with respect to the current
residual problem and the sampling threshold $\tau$.  Therefore, they need to be
updated when $\tau$ is decreased or 
when a new seed node is selected.  
While the new estimates can always be computed by simply scanning $\Index$,
this is inefficient.
In Section \ref{updatepps:sec} we present additional
structures which support efficient updates of samples and estimates.
\notinproc{ Finally,
Section \ref{alphaworstcase:sec} includes a
worst-case analysis.}

\subsection{PPS Samples of Influence Sets} \label{pps:sec}

  We start by specifying the sampling scheme, which is the core of
  our approach.
 The sample we maintain for each node $u$ is a 
{\em Probability Proportional to
  Size} (PPS) sample of all
node-instance pairs $(v,i)$, where the weighting is with respect to 
the contribution values  $\Delta^{(i)}_{uv}$, as defined in
Equation~\eqref{Deltadef:eq}.  
Recall from Equation~\eqref{minfdef} that the influence, which we are estimating from
the sample, is the sum of these contributions.
PPS sampling can be equivalently defined with respect to a threshold
$\tau$~\cite{SSW92}:  Each
entry $(v,i)$ has an independent $r^{(i)}_v \sim U[0,1]$ and it is included in
  the sample of $u$ if
\begin{equation} \label{ppssampcond}
\frac{\Delta^{(i)}_{uv}}{r^{(i)}_v} \geq \tau. 
\end{equation}

From the PPS sample we can unbiasedly estimate the influence of $u$,
using the classic inverse-probability estimator~\cite{HT52}.
We denote by $H_u$ the set of all pairs $(v,i)$ such that $\Delta^{(i)}_{uv}
\geq \tau$ and by $M_u$ the set of all other sampled pairs, that is,
those where $\Delta^{(i)}_{uv} \in [r^{(i)}_v \tau,\tau)$.
Note that pairs in $H_u$ are sampled with probability $1$, whereas pairs in $M_u$
are sampled with probability $\frac{\Delta^{(i)}(u)}{\tau}$ (this is
the probability of having a rank value so that \eqref{ppssampcond} is satisfied).
The estimate is the sum of the ratio of contribution to inclusion probability:
\begin{equation} \label{ppsinfest:eq}
\widehat{\INF}((\mathcal{G},\delta),u) = \frac{1}{\ell} \left(\tau|M_u|+ \sum_{(v,i) \in H_u}
\Delta^{(i)}_{uv}\right).
\end{equation}
With PPS sampling, when $\tau$ is low enough so that the estimate is
at least $k\tau$, which always happens when we have $k$ samples, 
the CV is at most 
$1/\sqrt{k}$. 
The estimate is also well concentrated according to the Chernoff
bound. 

We note that our use of PPS sampling, rather than uniform sampling, 
is critical to performance with general $\alpha$.  When using (for example)
exponential or polynomial decay, the positive contributions of different pairs to the influence of a
node can vary by orders of magnitude.  Therefore, 
we must use weighted sampling, where
heavier contributions are more likely to be sampled,  to obtain
good accuracy with a small sample size.
With threshold influence, in contrast, contributions were
either $0$ or $1$, which meant  that we could get good performance with
uniform sampling.

The PPS samples we will maintain for 
different nodes are computed with respect to 
the same threshold $\tau$. 
The samples are also
{\em coordinated}, meaning that the same values $r^{(i)}_v$ are used 
   in the samples of all nodes.  Coordination also helps us to
   maintain the samples when the contribution values $\Delta$ are
   modified, since the sample itself minimally changes to reflect the
   new values \cite{BrEaJo:1972,Ohlsson:2000,multiw:VLDB2009}.
In our implementation, node-instance pairs are assigned structured 
random permutation ranks, which are integers in $[n\ell]$, and for
permutation rank $h$ we use
$\Rank{v,i}\equiv r^{(i)}_v = h/(n\ell)$. 
\begin{table*}[t]
\centering
\small
\caption{Performance of TSKIM using~$k=64$, $\ell=64$, and exponentially distributed edge weights. We evaluate the influence on 512~(different) sampled instances for thresholds 0.1 and 0.01.}
\label{tab:tskim}
\begin{tabular}{@{}lrrrrrrrrrrrr@{}}
\toprule 
&&& \multicolumn{4}{c}{\textsc{Influence [\%]}} & \multicolumn{6}{c}{\textsc{Running Time [sec]}}\\
\cmidrule(lr){4-7}\cmidrule(l){8-13}
&&& \multicolumn{2}{c}{50 seeds} & \multicolumn{2}{c}{1000 seeds} & \multicolumn{2}{c}{50 seeds} & \multicolumn{2}{c}{1000 seeds} & \multicolumn{2}{c}{$n$~seeds} \\
\cmidrule(lr){4-5}\cmidrule(lr){6-7}\cmidrule(lr){8-9}\cmidrule(lr){10-11}\cmidrule(l){12-13}
instance & \#\,nodes & \#\,edges & 0.01 & 0.1 & 0.01 & 0.1 & 0.01 & 0.1 & 0.01 & 0.1 & 0.01 & 0.1\\
\midrule 
\notinarxiv{\instance{AstroPh} & 14,845 & 239,304 & 1.02 & 19.17 & 9.96 & 39.25 & 0.9 & 2.0 & 2.0 & 4.0 & 3.7 & 6.6\\
\instance{Epinions} & 75,888 & 508,837 & 0.53 & 8.52 & 2.88 & 12.68 & 2.0 & 5.2 & 6.3 & 11.1 & 14.1 & 21.3\\
\instance{Slashdot} & 77,360 & 828,161 & 0.72 & 19.97 & 3.90 & 25.04 & 1.9 & 14.6 & 7.6 & 27.9 & 18.9 & 40.5\\
\instance{Gowalla} & 196,591 & 1,900,654 & 0.62 & 14.13 & 1.93 & 17.61 & 4.4 & 21.8 & 14.8 & 36.9 & 47.6 & 81.7\\
\instance{TwitterF's} & 456,631 & 14,855,852 & 0.20 & 19.38 & 1.64 & 24.26 & 9.9 & 133.4 & 36.4 & 269.6 & 269.9 & 648.4\\
\instance{LiveJournal} & 4,847,571 & 68,475,391 & 0.07 & 9.16 & 0.33 & 13.81 & 34.6 & 606.0 & 117.5 & 1,244.4 & 1,983.4 & 4,553.9\\
\instance{Orkut} & 3,072,627 & 234,370,166 & 2.82 & 74.44 & 4.61 & 77.47 & 779.7 & 5,490.5 & 1,788.7 & 11,060.7 & 7,360.9 & 24,520.3\\
}
\onlyinarxiv{}
\bottomrule 
\end{tabular}
\end{table*}
\notinarxiv{
\begin{figure*}[tb]
\centering
\includegraphics[page=1]{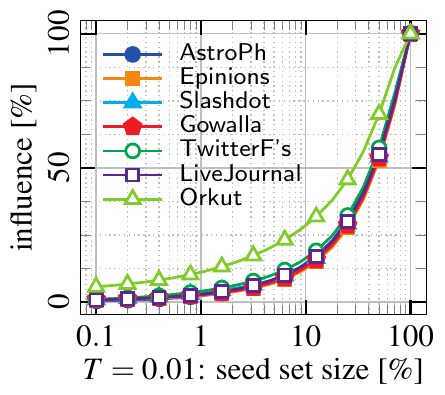}%
\includegraphics[page=3]{plots-timed/tskim_coverage.pdf}\hfill%
\includegraphics[page=2]{plots-timed/tskim_coverage.pdf}%
\includegraphics[page=4]{plots-timed/tskim_coverage.pdf}%
\caption{\label{fig:tskim001}\label{fig:tskim01}Evaluating influence permutations~(left) and running times~(right) on several instances for threshold decays 0.01 and 0.1. The legend applies to all plots.}
\end{figure*}
}
\onlyinarxiv{
\begin{figure*}[!t]
\centering
\includegraphics[page=1]{tskim_coverage.pdf}%
\includegraphics[page=3]{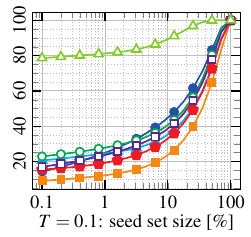}\hfill%
\includegraphics[page=2]{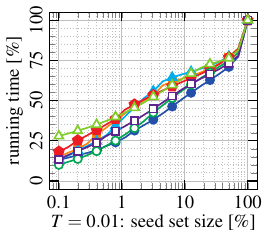}%
\includegraphics[page=4]{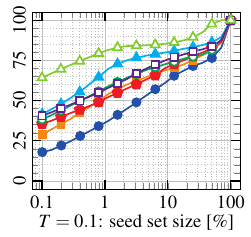}%
\caption{\label{fig:tskim001}\label{fig:tskim01}Evaluating influence permutations~(left) and running times~(right) on several instances for threshold decays 0.01 and 0.1. The legend applies to all plots.}
\end{figure*}}


\subsection{The Index Structure} \label{index:sec}
 The main structure we maintain is $\Index$, which can be viewed as an
inverted index of the PPS samples (but technically can include entries that
 used to be included in the PPS sample and may still be relevant).
For each node-instance pair $(v,i)$, $\Index{v,i}$ is
an ordered list of node-distance pairs $(u,d^{(i)}_{uv})$.  Note that
typically the lists are empty or very small for most pairs $(v,i)$.
The  list is ordered by 
increasing $d^{(i)}_{uv}$, which is the order in which a 
reverse Dijkstra algorithm
performed from $v$ on the
 graph $G^{(i)}$ (with  all edges reversed)
 scans new
 nodes. List $\Index{v,i}$ always stores (a prefix) of the scanned nodes,
 in scanning order.  It
 always includes 
 all nodes $u$ for which the pair 
$(v,i)$ is included in the PPS sample of $u$, that is, nodes $u$ that 
satisfy~\eqref{ppssampcond}. 
Each list $\Index{v,i}$ is trimmed from its tail so that it only contains
 entries $(u,d)$ where $\Delta^{(i)}_{uv} > 0$, that is,
$\alpha(d)- \alpha(\delta^{(i)}_v) > 0$. This is because other
entries have no contribution to the
marginal influence of $u$.
Note that the lists 
are always a prefix of the Dijkstra scan order, and once they 
are trimmed (from the end), they do not grow, and the respective
reverse Dijkstra computation never resumed,  even if $\tau$ decreases.

Each list $\Index{v,i}$  is logically viewed as having three
consecutive parts (that could be empty).
  The H part of the list are all entries $(u,d)$ with 
  $\alpha(d)-\alpha(\delta^{(i)}_v) \geq \tau$.  These are the nodes
  $u$ for which the pair $(v,i)$ contributes to the $H_u$ part of the
  PPS sample.
The M part of the list 
  are all entries with    $\alpha(d)-\alpha(\delta^{(i)}_v) \in
  [r^{(i)}_v\tau, \tau)$, which include nodes $u$ for which $(v,i)$
  contributes  to $M_u$.
 Finally, the L part of 
the list includes nodes for which
$0<\alpha(d)-\alpha(\delta^{(i)}_v) < r^{(i)}_v  \tau$.  The L nodes
do not currently include $(v,i)$ in the PPS sample of their influence
sets, but are still relevant, since this may change when $\tau$ decreases.

To support efficient updates of this classification,
we maintain $\HM{v,i}$ and $\ML{v,i}$, which contain the 
positions in the list $\Index{v,i}$ of the first $M$ and the first $L$ items (and are empty if
there are no M or L items, respectively).

 To efficiently compute the influence estimates, we maintain
for each node $u$ the values
\begin{eqnarray*}
\est.H[u] &=& \sum_{(v,i) \mid \Delta^{(i)}_{uv} \geq \tau}
\Delta^{(i)}_{uv} \\
\est.M[u] &=& | \{ (v,i) \mid  \Delta^{(i)}_{uv} \in [ r^{(i)}_v\tau,\tau)\}|
\end{eqnarray*}
The PPS estimate \eqref{ppsinfest:eq} on the influence of
$u$ is 
\begin{equation} \label{ppsalgest:eq}
\frac{1}{\ell}(\est.H[u]+ \tau \est.M[u])\ .
\end{equation}

\subsection{Reverse Dijkstra Computations} \label{rdijkstra:sec}

 We build the samples using reverse Dijkstra
  computations starting at node-instance pairs $(v,i)$.
The computation is from source $v$ in the transpose graph of
  $G^{(i)}$ and reveals all nodes $u$ for which the
  pair $(v,i)$ is included in the PPS sample for $u$ as defined in
  \eqref{ppssampcond}.
The nodes scanned by the reverse Dijkstra on $[v,i]$ are maintained as $\Index{v,i}$, in the same order. 
The computation for $(v,i)$ is 
paused once the distance $d$ from the source satisfies 
\begin{equation} \label{pauserule:eq}
\alpha(d)-\alpha(\delta^{(i)}_v) < r^{(i)}_v \tau. 
\end{equation}
The computation 
may resume when $\tau$ is decreased and the pause rule 
\eqref{pauserule:eq} no longer holds. 
It is not hard to verify that this pause rule suffices to obtain all entries of $(v,i)$ in 
the PPS samples of nodes. 
When the depth $d$ satisfies $\alpha(d)-\alpha(\delta^{(i)}_v) \leq
0$, the computation of the reverse Dijkstra $(v,i)$ is (permanently)
terminated, releasing all auxiliary data structures.
Note that the reverse Dijkstra computations for different pairs are
paused and resumed according to the global
threshold $\tau$, and can be performed concurrently.  

The algorithm
maintains ``state'' for all active Dijkstras.
We use the notation $\mu(v,i)=d$ for the next distance the reverse 
Dijkstra from $(v,i)$ would process when resumed.  Initially,
$\mu(v,i)=0$. 
  In order to efficiently determine the pairs $(v,i)$ for which
  reverse Dijkstra needs to be resumed, we maintain a 
max priority queue
$\Qgrand$ over 
  node-instance pairs $(v,i)$, prioritized by 
\begin{equation} \label{Qgrandpriority}
\frac{\alpha(\mu(v,i))-\alpha(\delta^{(i)}_v)}{r^{(i)}(v)}\ . 
\end{equation} 
This priority is the sampling threshold that is
  required to get $(v,i)$ into the PPS sample of the next 
node to be scanned by the reverse Dijkstra of $(v,i)$. 
We only need to resume the reverse Dijkstra $(v,i)$ when its priority
\eqref{Qgrandpriority} is at least $\tau$.  Note that the priority of
a pair $(v,i)$ can only decrease over time, when $\delta^{(i)}_v$ decreases or
when the reverse Dijkstra progresses and $\mu(v,i)$ increases.  This
allows us to maintain $\Qgrand$ with lazy updates.

In order to determine, after we decrease $\tau$, all the pairs for which the reverse 
Dijkstra computation should resume (or start), we simply extract 
all the top elements of  the queue $\Qgrand$ which have priority at least $\tau$.
These top elements are removed from $\Qgrand$.
The reverse Dijkstra is resumed on each removed pair $(v,i)$ until the
pause rule holds again, that is, we reach a distance $d$ such that 
$\alpha(d)-\alpha(\delta^{(i)}_v) < {r^{(i)}_v} \tau$.  At this  point
the reverse Dijkstra is terminated or paused.  If it is paused, we set 
 $\mu(v,i) \gets d$, and  the pair $(v,i)$ is placed in $\Qgrand$ with the new priority \eqref{Qgrandpriority}.

Note that the resume and pause rules of the reverse Dijkstras are 
  consistent with identifying all sampled pairs according to 
 \eqref{ppssampcond}, ensuring correctness.
\ignore{
The updated PPS influence estimates for all
nodes could always be computed by linearly scanning the current $\Index$ structure,
which stores the ordered lists of entries produced by the
  paused reverse Dijkstra computations. We can classify each entry as an H or an M entry
  and add its contribution to the estimated influence of the
  respective node.  The influence estimates change, however, when
  $\tau$ is decreased and when new seed nodes are added (and $\delta$
  values are decreased). In the sequel we
  introduce additional structures that allow us to do these updates efficiently.
}
 \notinproc{
\subsection{Selecting the Next Seed Node } \label{nextseed:sec}

  The algorithm decreases $\tau$  until we have sufficient confidence
  in the node with maximum estimated marginal influence.   The
  selection of the next node into the seed set is given in the
  pseudocode $\NextSeed$ in Appendix~\ref{app:ps:funcalpha}.

  We first discuss how we determine when we are happy with the maximum
  estimate.  When looking at a particular node, and 
the value of the estimate is at least  $\tau k$,  we know that the
value is well concentrated with CV that is $1/\sqrt{k}$.  We are,
however, looking at the maximum estimate among $n$ nodes.
To ensure an expected relative error of $\epsilon$ in the worst case, 
we need to apply a union bound and use $k=O(\epsilon^{-2}\log n)$.  The union bound ensures that the estimates
for all nodes have a well concentrated maximum error of $\epsilon$ times the maximum
influence.  In particular, the estimated maximum has a relative error
of $\epsilon$ with good concentration.

 In practice, however, the
influence distribution is skewed and therefore the union bound is too
pessimistic~\cite{binaryinfluence:CIKM2014}.  Instead, we propose the
following adaptive approach which yields tighter bounds for realistic instances.
Consider the node $u$ with maximum estimated marginal influence
$\hat{I}_u$  and let $I_u$ be its exact
marginal influence $I_u$.  The exact marginal influence can be compute,
using $\MargGain{u}$ (Section \ref{model:sec}), and is also computed anyway
when $u$ is added to the seed set.

The key to the adaptive approach is the following observation.
When working with a parameter $k=\epsilon^{-2}$ and the maximum
estimate is at least $k\tau$, then under-estimates  of the true
maximum are still well 
concentrated but over-estimates  can be large. Fortunately, however, over-estimates are 
easy to test for by comparing $I_u$ and $\hat{I_u}(1-\epsilon)$.

  We can use this as follows.   In our experiments 
we run the algorithm with a
  fixed $k$, always selecting the node with maximum estimate when the
  estimate exceeds $k\tau$.  We obtain, however, 
much tighter confidence interval than through the worst case
  bound. In particular, 
for the sequence of computed seeds, we track the sum $Er=\sum_u \max\{0,
(1-\epsilon)\hat{I}_u-I_u\}$.  This sum is added as a fixed component
to the confidence bound, which means that with high probability the
error for any prefix size is $Er$ plus a well concentrated value around $\epsilon$ times
the estimate.  In particular, the approximation ratio we obtain with
good probability for a prefix $s$ is at least $1-(1-1/s)^s -\epsilon- Er/\widehat{\INF}$.

Alternatively, (this is in the pseudocode) we can perform seed selection with respect to a
specified accuracy $\epsilon$.   Thus obtaining approximation ratio of
$1-(1-1/s)^s -\epsilon$ with good probability.
We use $k$ that is $\epsilon^{-2}$ but
when we have a candidate (a node with maximum estimated marginal gain which
exceeds $k\tau$) we  apply 
$\MargGain{u}$ to compute its exact marginal influence $I_u$.
If we find that $\hat{I}_u \geq (1-\epsilon) I_u$, we do not select
$u$ and instead decrease $\tau$ which returns to sampling.
Otherwise, we select $u$ into the seed set.

  We now consider tracking the node with
  maximum estimated influence.  To do that efficiently, we work with 
a max priority queue $\Qcands$, which contains nodes prioritized by their estimated 
influence, or more precisely, by $\est.H[u] + \tau \est.M[u]$.  
For efficiency, we use lazy updates, which allows priorities not to be
promptly updated when the estimate changes and instead, 
updated only for
elements which are the current maximum in the queue.

The function $\NextSeed$ repeats the following until a node is 
selected or $\perp$ is returned.  
If the maximum estimated priority in $\Qcands$ is less than $\tau k$, we return 
$\perp$. 
Otherwise, we remove from $\Qcands$
the node $u$ with maximum priority, compute $\hat{I}_u \gets \est.H[u]
+ \tau \est.M[u]$, and check if $\hat{I}_u$ is at lest the current 
maximum entry on the queue and is 
larger than $\tau k$. 
When using the first option, we simply return $(u,\hat{I}_u)$.  When working 
with a specified error, we 
further test $u$ as follows.  We compute the actual marginal 
 gain $I_u$ of $u$.  If it is lower than  $(1-\epsilon)\hat{I}_u$, we place $u$ back in $\Qcands$ with priority $I_u$
 and return $\perp$.  Otherwise, we return $u$.

  We now consider maintaining $\Qcands$.
Priorities are promptly updated to estimate values only when there is
a possibility of increase in the estimate.  This is necessary
for correctness of the priority queue when
using lazy updates.
The estimated influence of a node $u$ can only decrease 
when $\delta$ decreases for an element in the sample of $u$  (as a 
result of adding a new seed) or when 
$\tau$ decreases and no new entries are added to the sample.
Therefore, in these cases, we do not update the priorities promptly.
The priority of $u$ in $\Qcands$ is promptly updates when 
(as a result of decreasing $\tau$) a new 
entry is added to the sample of $u$.  

Note that our algorithm always maintains the estimation components $\est.H[u]$
and $\est.M[u]$ updated  for all nodes $u$, so the current estimate for
a node $u$ can at any time
be quickly computed from these two values and the current $\tau$.

 Lastly, as in the pseudocode, Algorithm \ref{alphaskim:alg} is executed
 to exhaustion, until the seed set influences all nodes. 
  The algorithm can be stopped also when the seed set $S$ is of
  certain desired size or  when a desired coverage
  $\sum_{u\in S}  I_u/(n\ell\alpha(0))$ is achieved.

} 

\subsection{Updating PPS Estimate Components} \label{updatepps:sec}

 The positions  $\HM{v,i}$ and 
  $\ML{v,i}$, and accordingly  the classification of entries $(u,d)$
  as H,M, or L,
 can be updated both when $\tau$ or $\delta^{(i)}_v$ decreases.  
When $\tau$ decreases, entries can only 
  ``move up:''  L entries can become M or H entries and M entries can 
  become H entries.  New entries can also be generated by a reverse
  Dijkstra on $(v,i)$.  Newly generated entries are always H or M entries.
  When $\delta^{(i)}_v$ decreases, entries can 
  ``move down.''  In addition, entries at the tail of $\Index{v,i}$,
  those with $\alpha(d)  \leq \alpha(\delta^{(i)}_v)$, get removed. 

  When an entry $(u,d)$ changes its classification, or when a new
  entry is generated by a reverse Dijkstra, we may need to update
the estimate components $\est.H[u]$ and $\est.M[u]$.

\notinproc{\subsubsection{Initial Updates When $\tau$ Decreases}}
\onlyinproc{{\noindent {\bf Initial Updates When $\tau$ Decreases}}}
  When $\tau$ decreases, we first (before resuming the reverse
  Dijkstra's) need to update the classification of existing entries and the implied 
  changes on $\est$.  \notinproc{The pseudocode for this update is provided as 
  the function $\MoveUp{}$ in Appendix~\ref{app:ps:funcalpha}.}

  To efficiently identify the $\Index$ lists that have entries that change their classification,
we maintain a max priority queues $\Qhml$. It contains node-instance pairs with priority equal to the 
\emph{reclassification threshold}, the highest $\tau$ that would require 
reclassification of at least one entry in the list.  
The procedure $\UpdateReclassThresh$ computes the reclassification
threshold for a pair $(v,i)$ and places the pair with this priority in $\Qhml$.
When $\tau$ is decreased, we only need to process lists of pairs that
are at the top of the queue
$\Qhml$. 

Lastly, we discuss the processing of a list $(v,i)$ which requires
reclassification.
The reclassification, the updates
of the estimation components $\est$, and the update of the
reclassification threshold using
$\UpdateReclassThresh{v,i}$,  are all
performed in computation that is proportional 
to the number of reclassified entries.  In particular, processing does not require
scanning the full list $\Index{v,i}$. This is enabled by the
pointers $\HM{v,i}$ and $\ML{v,i}$.

\notinproc{\subsubsection{New Scanned Node}}
\onlyinproc{{\noindent {\bf New Scanned Node}}}
The estimation components also need to be updated when a new entry
is appended to the $\Index{v,i}$ list when running a
reverse Dijkstra for $(v,i)$. The pseudocode for this update is included in
Algorithm~\ref{alphaskim:alg}.   The new scanned node $u$ with
distance $d$ creates a new entry $(u,d)$ which is appended
to the end of $\Index{v,i}$.  
A new entry is always an H or M entry (otherwise the pause rule applies).
If $\HM{v,i} \not=\perp$, that is, there is at least one M entry
in $\Index{v,i}$, the new entry must also be an M entry.  In this case,
we increase $\est.M[u]$ by $1$.  If $\HM{v,i}=\perp$, we check if
$(c \gets \alpha(d)-\alpha(\delta^{(i)}_v)) \geq \tau$.  If so, the new entry
is an H entry and we increase $\est.H[u]$ by $c$.  Otherwise, it is
the first M entry.  We set
$\HM{v,i} \gets |\Index{v,i}|-1$, $\est.M[u]=1$, and 
insert the pair $(v,i)$ to the queue $\Qhml$ with priority $c$.
After updates are completed, we recompute the estimated influence
$\est.H[u]+\tau\est.M[u]$ of the node $u$ and update accordingly the
priority of $u$ in $\Qcands$.

\notinproc{\subsubsection{New Seed Node}}
\onlyinproc{{\noindent {\bf New Seed Node}}}
  When a new seed node $u$ is selected, 
we perform a forward Dijkstra from the seed in each instance $i$.  We update $\delta^{(i)}_v$ at
  visited nodes $v$ and compute the exact marginal influence of the new
  seed.  
The forward Dijkstra is pruned at nodes $v$  with $\delta^{(i)}_v$ that is
smaller or equal to their distance from $u$.  When we update 
$\delta^{(i)}_v$, we also may need to reclassify entries
in $\Index{v,i}$, update the positions $\HM{v,i}$ and $\ML{v,i}$ 
 and update estimation components $\est.H[u]$ and $\est.M[u]$ of reclassified
entries $(u,d)$ .  \onlyinproc{We refer to this update as the function
  $\MoveDown()$. Pseudocode is provided in the full version.}\notinproc{A pseudocode for this
update is provided as the function $\MoveDown()$ in Appendix~\ref{app:ps:funcalpha}.}
  We also update the priority of the pair
$(v,i)$ in $\Qhml$ and decrease its priority in $\Qgrand$ to
$(\alpha(\mu(v,i))-\alpha(\delta^{(i)}_v))/r^{(i)}_v$ (since we do not track
$\mu(v,i)$ explicitly in the pseudocode, we instead decrease the
priority to reflect the decrease in $\delta^{(i)}_v$). If the
updated  priority in $\Qgrand$ is
$\leq 0$, the reverse Dijkstra of $(v,i)$ is terminated and it is
removed from $\Qgrand$.
Note that in this update, entries can only be reclassified ``down:''
E.g. an entry $(u,d)$ that was in $H$ can move to
$M$, $L$, or be purged, if $\alpha(d) \leq \alpha(\delta^{(i)}_v)$.

\ignore{
\subsection{Implementation Comments}
 A simple acceleration we use in our implementation is to replace
$\alpha(x)$ by $I_{\alpha(x) > \epsilon 
  \alpha(0)} \alpha(x)$, where $I$ is the indicator function.  That 
is, setting $\alpha(x)$ to $0$ when 
$\alpha(x)\leq \epsilon \alpha(0)$.   This reduces the number of
distance updates which do not contribute much to the influence
computation.
}

\notinproc{
\subsection{Analysis} \label{alphaworstcase:sec}

When we run the algorithm with fixed
$k=O(\epsilon^{-2} \log n)$\notinproc{ or use the adaptive approach in seed
selection (as detailed in Section \ref{nextseed:sec})}, we have the
following guarantee on the approximation quality:
\begin{theorem} \label{alphaaccuracy:thm}
 \alphaskim\ returns a
sequence of seeds so that for each prefix $S$ of size $s$, with high
probability,
\begin{equation} \label{approxguarantee}
\INF(\mathcal{G},S) \geq (1-(1-1/s)^s -\epsilon) \max_{ U \mid |U|= s}
\INF(\mathcal{G},U).
\end{equation}
\end{theorem}
\notinproc{
\begin{proof}
The algorithm, with very high probability, selects a node with
marginal influence that is at least $1-\epsilon$ of the maximum one.
This follows from a union bound over all steps and nodes of
the quality of the estimate obtained from a PPS sample.
 We then apply an approximate variant (e.g., \cite{binaryinfluence:CIKM2014}) of the
classic proof \cite{submodularGreedy:1978}  of the approximation ratio of {\sc Greedy} for
monotone submodular problems.
\end{proof}
}

  The worst-case bound on the (expected) running time of
  \alphaskim is as follows.\notinproc{The proof is provided in Appendix
  \ref{askimtime:proofs}.}
\notinproc{ We  note that our design exploits properties of real 
  instances and the poly-logarithmic overheads are not observed in experiments.}
\begin{theorem} \label{alphatime:thm}
\alphaskim\ runs in expected time 
$$O\left(\frac{\log^3 n}{\epsilon}\sum_{i=1}^\ell |E^{(i)}| + \frac{\log^3
  n}{\epsilon^3}n+ \frac{\log n}{\epsilon^2}|\bigcup_{i=1}^\ell E^{(i)}|\right)$$
providing the guarantee \eqref{approxguarantee}.
\end{theorem}
 } 

\ignore{
\section{Adaptive confidence bounds} \label{confidence:sec}

{\bf We start working here with what we proposed in
  \cite{binaryinfluence:CIKM2014}.  It requires putting more things in
  the code, like tracking the size of the second largest sketch.  We
  work with that and see estimates accuracy as seeds are added.
  After fine tuning there,  we cook up strategies of
  changing $k$ parameter on the go. }
}

\section{Experiments} \label{sec:exp}

Our algorithms were implemented in C++ and compiled using Visual Studio 2013 with full optimization. Our test machine runs Windows 2008R2 Server and has two Intel Xeon E5-2690 CPUs and 384\,GiB of DDR3-1066 RAM. Each CPU has 8 cores~(2.90\,GHz, 8\,$\times$\,64\,kiB L1, 8~$\times$\,256\,kiB, and 20\,MiB L3 cache). For consistency, all runs are sequential.

\notinproc{
The datasets in our experiments are obtained from the SNAP~\cite{SNAP} project and represent \emph{social}~(\instance{Epinions}
, \instance{Slashdot}~\cite{lldm-cslnn-09}%
, \instance{Gowalla}~\cite{cho2011friendship}%
, \instance{TwitterFollowers}~\cite{DLMM:SR13}%
, \instance{LiveJournal}~\cite{BHKL:KDD2006}%
, \instance{Orkut}~\cite{YL:ICDM2012})
and \emph{collaboration}~(\instance{AstroPh}~\cite{LeskovecKF:TKDD07})
networks. All these graphs are unweighted. 
}

Unless otherwise mentioned, we 
test our algorithms
using~$\ell = 64$ independent instances generated from the graph
by assigning independent random length to every edge according to an exponential 
distribution with expected value~1 \cite{Gomez-RodriguezBS:ICML2011,ACKP:KDD2013,CDFGGW:COSN2013}. \notinproc{To do so, we sample a value~$x$ uniformly at random from the range $(0,1]$, then set the edge length to~$-\ln x$.}
We use ADS parameter~$k = 64$.

\subsection{Distance-Based Influence Maximization}

We start with the Influence Maximization problem. Recall that we consider two variants of this problem: threshold influence and general distance-based influence. We discuss each in turn. 

\subsubsection{Threshold Influence}

\ignore{
For threshold influence and some threshold~$T$, we set~$\alpha(x) = 1$ for~$x \leq T$ and $\alpha(x) = 0$ otherwise. We consider $T = 0.01$ and $T = 0.1$ in our experiments. 
}

Our first experiment considers the performance of \tskim\
(\notinproc{Algorithm~\ref{Tskim:alg}, }Section \ref{thresh:sec}),
which finds a sequence of seed nodes such that each prefix of the
sequence approximately maximizes the influence. Our results are
summarized in Table~\ref{tab:tskim}. For each dataset, we first
report its total numbers of nodes and edges. This is followed by the
total influence~(as a percentage of the total number of nodes of the
graph) of the seed set found by our algorithm. We report figures for
50 and 1000 seeds~and for threshold values $T=0.01$ and $T=0.1$. Finally, we
show the total running time of our algorithm when  it is stopped after
computing an approximate greedy sequence of 50, 1000, or
\emph{all~$n$ nodes}.
 Note that we omit the respective influence figure for the seed set
 that contains all nodes,  since it is~100\% by definition.

The table shows that, unsurprisingly, the higher threshold has higher
influence values.  This is because the coverage function is monotone
non-decreasing in $T$.   The running time of our algorithm depends on that influence~(since its graph searches must run for longer), but it is still practical even for fairly large thresholds and even if we compute the entire permutation. For the largest graph we test~(\instance{Orkut}), with hundreds of millions of edges, we can compute the top 50~seeds in less than 15~minutes, and order all nodes in a few hours using a single CPU core.

Figure~\ref{fig:tskim001} presents a more detailed perspective on the same experiment. It shows, for~$T = 0.01$ and~$T = 0.1$, how total influence and the running times depend on the size of the seed set. We note that the first few seeds contribute with a disproportionate fraction of the total influence, particularly with~$T = 0.1$, and an even higher percentage of the total running time. The overall shape of the curves is quite similar, with \instance{Orkut} as a noticeable outlier: its first few seeds contribute relatively more to the overall influence than in other instances. Note that~\instance{Orkut} is also the densest instance in our testbed.

\notinarxiv{
\begin{figure}
\includegraphics{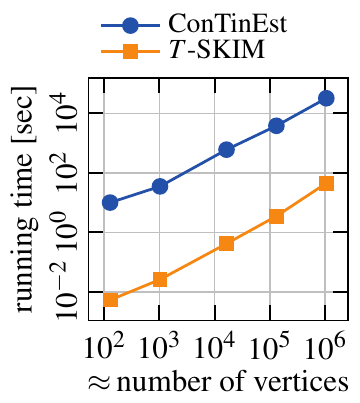}%
\hfill 
\includegraphics[page=1]{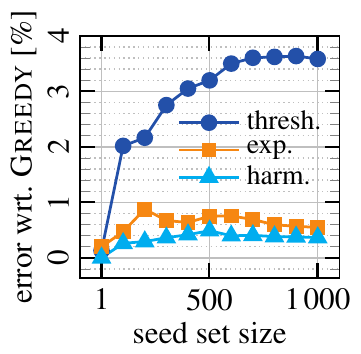}%
\caption{\label{fig:nips_and_error}Left: Comparing \tskim\ to ConTinEst. Right: Error of \tskim\ and \alphaskim.}%
\end{figure}
}
\onlyinarxiv{
\begin{figure}
\includegraphics{nips}%
\hfill 
\includegraphics[page=1]{comparison}%
\caption{\label{fig:nips_and_error}Left: Comparing \tskim\ to ConTinEst. Right: Error of \tskim\ and \alphaskim.}%
\end{figure}
}

We now compare \tskim\ to ConTinEst, the algorithm by Du et al.~\cite{DSGZ:nips2013}. Although their sequential implementation is publicly available, we were unable to run it on our inputs within reasonable time. (A preliminary test on \instance{AstroPh}, our smallest instance, did not produce any output within five hours.) Note that to evaluate graphs with more than 1024 vertices, they actually use a distributed implementation, which they run on a cluster with 192 cores. Unfortunately, we had access neither to such a cluster nor to the distributed implementation of their algorithm.

In order to still be able to make some comparison, we generated the same instances as they used in their evaluation: core-periphery Kronecker networks~\cite{SNAP}~(parameter matrix: [0.9 0.5; 0.5 0.3]) of varying size, using the Weibull distribution \notinproc{$$f(x, \lambda, \beta) = \frac{\beta}{\lambda} \Big(\frac{x}{\lambda}\Big)^{\beta-1} e^{-(x/\lambda)^\beta}$$} for the edge lengths~\cite{lawless2011statistical}. (Note that~$\lambda$ controls scale and~$\beta$ shape.) For each edge we chose~$\lambda$ and~$\beta$ uniformly at random from~$(0,10]$. We ran the same experiment as they did, setting~$|S| = 10$, and $T = 10$. Figure~\ref{fig:nips_and_error} (left) shows the running times for Kronecker networks of varying size. We observe that our approach run on a single CPU core is consistently about 3 orders of magnitude faster than their algorithm run on a cluster. Unfortunately, we were not able to compare the computed influence, as those figures are not reported in~\cite{DSGZ:nips2013}.

\subsubsection{General Distance-Based Influence}

We now evaluate \alphaskim, a more general version of our IM algorithm that can handle arbitrary decay functions. For this experiment, we consider both harmonic and exponential decay functions, the most commonly used in the literature. To test harmonic decay, we use~$\alpha(x) = 1 / (10x + 1)$; for exponential decay, we use~$\alpha = e^{-10x}$. These functions turn out to give interesting influence profiles. In \alphaskim\ we initialize~$\tau$ to~$n\ell/k$ and set~$\lambda$ to~0.5.

Table~\ref{tab:askim} shows, for both functions, the influence values~(in percent) obtained by \alphaskim\ for 50~and 1000~seeds, as well as the corresponding running times.

\begin{table}[!h]
\centering
\caption{Performance of $\alpha$-SKIM using~$k=64$,~$\ell=64$, and exponentially distributed edge weights for 50 and 1000 seeds. We use exponential~(exp.:~$\alpha\colon x \mapsto e^{-10x}$) and harmonic~(harm.:~$\alpha\colon x \mapsto 1/(10x+1)$) decay functions.}\label{tab:askim}
 \small
\setlength{\tabcolsep}{.9ex}
\begin{tabular}{@{}lrrrrrrrr@{}}
\toprule
 & \multicolumn{4}{c}{\textsc{Influence~[\%]}} & \multicolumn{4}{c}{\textsc{Running time~[sec]}}\\
\cmidrule(lr){2-5}\cmidrule(l){6-9}
& \multicolumn{2}{c}{50 seeds} & \multicolumn{2}{c}{1000 seeds} & \multicolumn{2}{c}{50 seeds} & \multicolumn{2}{c}{1000 seeds}\\ 
\cmidrule(lr){2-3}\cmidrule(lr){4-5}\cmidrule(lr){6-7}\cmidrule(l){8-9}
instance&exp.&harm.&exp.&harm.&exp.&harm.&exp.&harm.\\
\midrule
\notinarxiv{\instance{AstroPh} & 17.6 & 31.4 & 33.5 & 44.9 & 15 & 15 & 43 & 40\\
\instance{Epinions} & 7.6 & 14.9 & 11.2 & 18.2 & 35 & 40 & 93 & 99\\
\instance{Slashdot} & 16.9 & 29.1 & 21.3 & 32.8 & 104 & 88 & 238 & 224\\
\instance{Gowalla} & 13.1 & 25.9 & 15.9 & 28.2 & 166 & 213 & 323 & 455\\
\instance{TwitterF's} & 16.0 & 26.3 & 19.7 & 29.2 & 1,500 & 1,387 & 2,459 & 2,816\\
\instance{LiveJournal} & 10.6 & 23.5 & 13.4 & 25.8 & 5,637 & 7,765 & 11,906 & 13,016\\
}
\onlyinarxiv{}
\bottomrule
\end{tabular}
\end{table}

The table shows that \alphaskim\ is slower than \tskim\ by up to an order of magnitude for comparable influence. In fact, if we ran \alphaskim\ with a threshold function~(not shown in the table), it would be about three times as slow as \tskim, while producing the exact same results. However, this is to be expected, since \alphaskim\ is a much more sophisticated~(and flexible) algorithm, which, unlike \tskim, can handle smooth decay functions with guarantees.

\notinarxiv{
\begin{figure*}
\centering
\includegraphics[page=1]{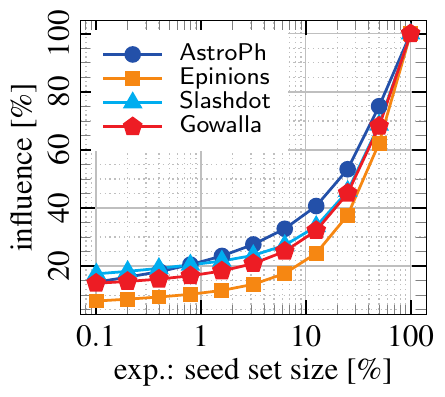}%
\includegraphics[page=3]{plots-timed/askim_coverage.pdf}\hfill%
\includegraphics[page=2]{plots-timed/askim_coverage.pdf}%
\includegraphics[page=4]{plots-timed/askim_coverage.pdf}%
\caption{\label{fig:askim}Evaluating influence permutations~(left) and running time~(right) on several instances for exponential~(exp.: $\alpha\colon x \mapsto e^{-10x}$) and harmonic~(harm.: $\alpha\colon x \mapsto: 1/(10x+1)$) decays. The legend applies to all plots.}
\end{figure*}
}
\onlyinarxiv{
\begin{figure*}
\centering
\includegraphics[page=1]{askim_coverage.pdf}%
\includegraphics[page=3]{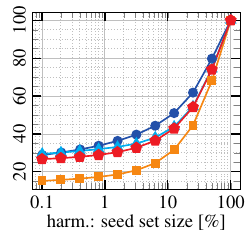}\hfill%
\includegraphics[page=2]{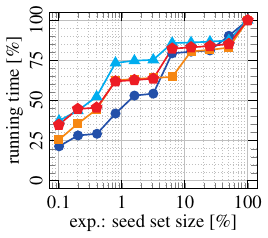}%
\includegraphics[page=4]{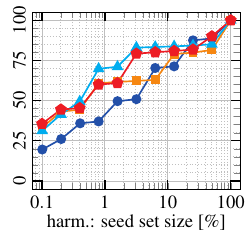}%
\caption{\label{fig:askim}Evaluating influence permutations~(left) and running time~(right) on several instances for exponential~(exp.: $\alpha\colon x \mapsto e^{-10x}$) and harmonic~(harm.: $\alpha\colon x \mapsto: 1/(10x+1)$) decays. The legend applies to all plots.}
\end{figure*}
}
\notinarxiv{
\begin{figure*}[t]
\centering 
\includegraphics[page=1]{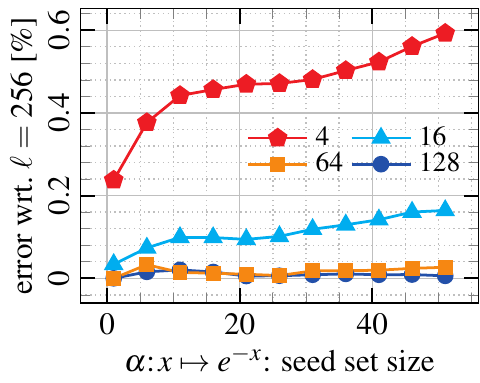}\hfil%
\includegraphics[page=2]{plots-timed/ell.pdf}\hfil%
\includegraphics[page=3]{plots-timed/ell.pdf}\hfil 
\caption{Evaluating different numbers of simulations~($\ell$-values) for different decay functions on \instance{AstroPh}.}\label{fig:ell}
\end{figure*}
}
\onlyinarxiv{
\begin{figure*}[t]
\centering 
\includegraphics[page=1]{ell.pdf}\hfil%
\includegraphics[page=2]{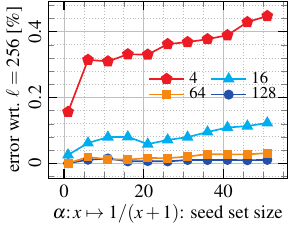}\hfil%
\includegraphics[page=3]{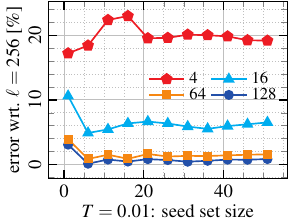}\hfil 
\caption{Evaluating different numbers of simulations~($\ell$-values) for different decay functions on \instance{AstroPh}.}\label{fig:ell}
\end{figure*}
}
\begin{table*} [h!]
\centering
\caption{Evaluating the distance-based influence oracle with~$\ell=64$.}\label{tab:toracle}
\setlength{\tabcolsep}{.65ex}
\small
\begin{tabular}{@{}lrrrrrrrrrrrrrrrrrrrr@{}}
\toprule
&\multicolumn{2}{c}{\textsc{Preprocessing}}&\multicolumn{6}{c}{\textsc{Queries with~$\alpha \colon x \mapsto e^{-10x}$}}&\multicolumn{6}{c}{\textsc{Queries with~$\alpha \colon x \mapsto 1/(10x+1)$}}&\multicolumn{6}{c}{\textsc{Queries with~$T=0.01$}}\\
\cmidrule(lr){2-3}\cmidrule(lr){4-9}\cmidrule(lr){10-15}\cmidrule(l){16-21}
&&&\multicolumn{2}{c}{1 seed}&\multicolumn{2}{c}{50 seeds}&\multicolumn{2}{c}{1000 seeds}&\multicolumn{2}{c}{1 seed}&\multicolumn{2}{c}{50 seeds}&\multicolumn{2}{c}{1000 seeds}&\multicolumn{2}{c}{1 seed}&\multicolumn{2}{c}{50 seeds}&\multicolumn{2}{c}{1000 seeds}\\
\cmidrule(lr){4-5}\cmidrule(lr){6-7}\cmidrule(lr){8-9}\cmidrule(lr){10-11}\cmidrule(lr){12-13}\cmidrule(lr){14-15}\cmidrule(lr){16-17}\cmidrule(lr){18-19}\cmidrule(l){20-21}
& time & space & time & err. & time & err. & time & err. & time & err. & time & err. & time & err. & time & err. & time & err. & time & err.\\
instance & [h:m] & [MiB] & [\textmu{}s] & [\%] & [\textmu{}s] & [\%] & [\textmu{}s] & [\%] & [\textmu{}s] & [\%] & [\textmu{}s] & [\%] & [\textmu{}s] & [\%] & [\textmu{}s] & [\%] & [\textmu{}s] & [\%] & [\textmu{}s] & [\%]\\
\midrule
\notinarxiv{\instance{AstroPh} & 0:10 & 149.2 & 38 & 7.2 & 9,695 & 1.2 & 229,340 & 0.5 & 31 & 4.4 & 9,152 & 4.1 & 227,943 & 0.5 & 27 & 1.1 & 8,855 & 0.4 & 204,551 & 2.8\\
\instance{Epinions} & 0:46 & 674.0 & 32 & 3.2 & 8,552 & 1.1 & 222,470 & 1.0 & 26 & 2.2 & 9,203 & 1.2 & 196,717 & 0.5 & 22 & 0.5 & 8,267 & 0.3 & 191,709 & 0.6\\
\instance{Slashdot} & 1:10 & 851.4 & 46 & 5.6 & 11,884 & 1.5 & 310,170 & 0.4 & 38 & 3.2 & 10,970 & 1.9 & 291,185 & 1.2 & 73 & 0.6 & 13,768 & 0.4 & 247,509 & 0.6\\
\instance{Gowalla} & 3:55 & 2,558.6 & 52 & 3.8 & 17,109 & 1.0 & 356,818 & 0.4 & 47 & 2.9 & 14,151 & 2.2 & 289,318 & 0.8 & 61 & 1.2 & 16,092 & 0.6 & 329,976 & 0.3\\
\instance{TwitterF's} & 19:33 & 6,165.1 & 51 & 3.8 & 13,816 & 1.4 & 365,366 & 0.7 & 42 & 2.6 & 13,166 & 1.5 & 379,296 & 0.9 & 39 & 2.3 & 13,912 & 0.7 & 360,766 & 0.2\\
}
\onlyinarxiv{}
\bottomrule
\end{tabular}
\end{table*}

Even though \alphaskim\ is slower, it is still practical. It scales well with the number of seeds (increasing from 50 to 1000 barely doubles the total running time) and can still handle very large graphs. 

Figure~\ref{fig:askim} presents a more detailed view of the same experiment (for a few graphs), with up to $n$ seeds. It shows that computing a full permutation (with $n$ seeds) is not much more expensive than computing $n /1000$ (a few dozen) seeds. An interesting difference between these results and those for \tskim\ (reported in Figure~\ref{fig:tskim001}) is that for \alphaskim\ the running time grows less smoothly with the number of seeds. The discontinuities correspond to decreases in the sampling threshold $\tau$, causing additional sampling.

\subsubsection{Solution Quality}

Figure~\ref{fig:nips_and_error} (right) compares the quality of the
seed sets found by~\tskim~(for threshold decay) and \alphaskim~(for
exponential and harmonic decays) with those found by exact \textsc{Greedy}
on \instance{AstroPh} ($\ell=64$ simulations). We consider sets of size 1 to 10$^3$ and the same decay functions as above. Each point of the curve represents the error~(in percent) of our algorithm when compared to \textsc{Greedy}. We observe that the error is very low in general~(less than 1\% for exponential and harmonic decay, and less than 4\% for threshold). Considering the fact that \skim\ is many orders of magnitude faster than \textsc{Greedy}~(while still providing strong guarantees), these errors are acceptable. Note that the error of the first seed vertex is very low in all cases~(close to 0\%), indicating that \skim\ does very well in finding the most influential node.

The quality of the solutions provided by the algorithm with respect to 
the probabilistic input (graph distribution) depends on the number of instances (simulations) $\ell$. Our experiments so far have used $\ell = 64$. We now compare this with other choices of $\ell$. 
Figure~\ref{fig:ell} compares the quality of the seed sets found by {\sc Greedy} for \instance{AstroPh} for $\ell = 4, 16, 64, 128$ with those found by $\ell = 256$. We consider sets of size 1 to 50 and three different decay functions: exponential, harmonic, and threshold (with $T = 0.01$). Each point in the curve represents the error (in percent) relative to the solution with $\ell = 256$. Although the error is consistently high for the threshold IM when $\ell$ is very small, it becomes negligible for $\ell \geq 64$, justifying our choice of parameters. For smoother (exponential or harmonic) decay, all errors are significantly smaller, and even smaller values of $\ell$ would be acceptable.

\subsection{Distance-Based Influence Oracle}

We now evaluate our influence oracles. Recall that this setting has two stages. The preprocessing stage takes as input only the graph and computes sketches. The query stage takes a set~$S$ of seeds and a function~$\alpha$ and uses the sketches to estimate the influence of~$S$ with respect to~$\alpha$. Note that same preprocessing stage can be used to answer queries for any decay function~$\alpha$. For this experiment, we consider three such functions: exponential~($\alpha(x) = e^{-10x}$), harmonic~($\alpha(x) = 1 / (10x + 1)$), and threshold (with $T = 0.01$).

Table~\ref{tab:toracle} summarizes our results in this setting. For each dataset tested, it first shows the preprocessing time and the total space required to store all sketches. Then, for each decay function, we report the query time (in microseconds) and the estimation error for random sets $S$ of sizes 1, 50, and 1000. (Note that measuring the error requires computing exact influence of each seed set with multiple Dijkstra searches; this time is not included in the table.) Each entry in the table is the average of 100 random seed sets. 

The table shows that, as predicted, query times are almost independent of the $\alpha$ function, the size of the influenced set, and the size of the graph. Moreover, they have a slightly superlinear dependence on the number of seeds. Queries are somewhat slower than for reachability-based IC~(as reported in~\cite{binaryinfluence:CIKM2014}), since sketches are bigger and the estimator is more involved. Our oracles are much more flexible, however, and still practical. For 50 seeds, one can answer queries in a few milliseconds, whereas an exact computation could take minutes or more on large graphs. Moreover, its error is consistently low, regardless of the number of seeds.

\section{Conclusion} \label{conclu:sec}
We formulate and study a 
model of {\em distance-based influence} which captures decay of
relevance with distance, and 
unifies and extends several established research threads.
We design novel near-linear algorithms for greedy influence
maximization and a construction of influence oracles.
 In future, we hope to apply our novel
weighted-sampling approach for approximate greedy influence
maximization to scale up other submodular maximization problems.
\ignore{
We provide the first 
scalable algorithms for influence queries and influence
maximization.
Our approach is the first to work with natural smooth decay functions
including exponential and polynomial decay.  For the threshold model,
which was previously studied, we design much more scalable algorithms.
Our algorithms provide statistical guarantees on approximation
quality, theoretical near-linear worst-case bounds on running time,
and are demonstrated to scale to large networks.
}
\paragraph*{Acknowledgements} We would like to thank the authors of~\cite{DSGZ:nips2013} for helping us reproduce their inputs and pointing us to their implementation of ConTinEst.

{\small
\bibliographystyle{plain}
\bibliography{cycle,dblp} 
}

\onlyinproc{\end{document}}
\clearpage
\appendix

\section{Running time analysis for \tskim} \label{tskimtime:sec}

The computation is dominated by two 
components.  The first is the reversed pruned Dijkstra searches which
build the inverted sketches.  The second is the forward pruned Dijkstra searches that we use to compute
and update the residual problem after a new seed node is selected.
 In both components we use pruned shortest-paths computations (whereas
 \skim\ for reachability-based influence uses generic searches). The bound on the number of
reversed edge traversals performed for sketch building and the
analysis deriving this bound are essentially the same as with
reachability-based \skim\ \cite{binaryinfluence:CIKM2014}:  We bound the total 
number of increments to the entries in the array $\Size$.  Each entry
in that array corresponds to a node and the size (number of entries)
in the partial sketch.  Recall that the sketches themselves are maintained in an inverted structure.

 The number of forward traversals from $u$ in instance $i$ corresponds to the number of times 
the distance $\Covered{u,i}$ of node $u$ in instance $i$ from the seed
set is updated.
 As we discussed in Section \ref{approxgreedy:sec}
the worst-case number of updates can depend on the number
of seeds with exact distance-based greedy.  The approximate selection of
\tskim, however, is such that all seed nodes with marginal influence
that is close to the maximum one (within $k^{-0.5}$) have
approximately the same probability of being selected.  This means that
we can invoke Theorem \ref{randupdatebound:thm} to obtain an expected
$O(\epsilon^{-1}\log^2 n)$ bound.

  We therefore obtain the following time bound for an exhaustive
  execution of \tskim\ (until influence of $S$ approaches $n$):
\begin{theorem}
\tskim\ guarantees with high
probability, for all $s$, that the first $s$ selected seeds have influence 
that is at least $1-(1-1/s)^s
-\epsilon$ times the maximum influence of a seed set of size $s$, and
runs in expected time
\begin{equation} \label{tskimtime:eq}
O(\epsilon^{-1}\sum_i 
|E^{(i)}|\log^3 n + m\epsilon^{-2} \log^3 n)\ ,
\end{equation}
 where $m\leq |\bigcup_i 
E^{(i)}|$ is the sum over nodes of the maximum in-degree over 
instances. 
\end{theorem}

 We note that this upper bounds the worst-case behavior.
  Our experiments verify that typically the computation is dominated
  by the reverse searches and the number of edge traversals is much smaller.
  We propose another  solution of theoretical interest  which provides worst-case robustness of
  the number of forward traversals at the cost of a slight
  softening of the sharp threshold.  This design
  was not implemented.

We first formalize the notion of a 
softened threshold $T$ in the influence function: For a
parameter 
$\nu\in (0,1)$, we only require approximation with respect to the
maximum influence of a seed set of size $s$ {\em when
the threshold is $(1-\nu)T$}.  Since we expect the value of the
threshold used in an application not be precise anyway, this relaxation
approximates the intended semantics of using a
threshold of $T$.

  We next outline a slight modification of \tskim\ which provides a bound
  on the running time and approximation quality with respect to the soft threshold.
 \begin{theorem}
Our modified \tskim, guarantees with high
probability, for all $s$ that the first $s$ selected seeds have influence for
threshold function $\alpha(x) = x\leq T \, ? \, 1
\,  : \, 0$ that is 
at least $1-(1-1/s)^s
-\epsilon$ times the maximum influence of a seed set of size $s$ with
respect to a threshold function $\alpha(x) = x\leq (1-\nu)T \, ? \, 1
\,  : \, 0$.
The algorithm runs in time
\begin{equation} \label{tskimsoft:eq}
O(\nu^{-1}\sum_i 
|E^{(i)}|\log n + m\epsilon^{-2} \log^3 n)\ ,
\end{equation}
 where $m\leq |\bigcup_i 
E^{(i)}|$ is the sum over nodes of the maximum in-degree over 
instances. 
\end{theorem}
\begin{proof}
We outline the modifications needed to support the soft threshold.

In the forward searches, we only propagate updates to $\Covered{v,i}$
when the decrease in distance  is at least $\nu T$.  In particular,
the reversed Dijkstras are always pruned at distance $T(1-\nu)$
and the forward Dijkstras are pruned when $d\leq T$ or $d>\Covered{v,i} -\nu T$. 
This means that at any given time, if $d^{(i)}_{Sv} \leq (1-\nu T)$
then $\Covered{v,i} \in [d^{(i)}_{Sv},d^{(i)}_{Sv}+\nu T]$.
That is, node-instance pairs of distance at most 
$(1-\nu)T$ are counted as influenced and pairs of distance greater than $T$ are not influenced, but other pairs can be counted either way and our estimation guarantee is with respect to a threshold 
of $(1-\nu)T$. 

 When discretizing to $\nu T$, we obtain a worst-case guarantee of 
$1/\nu$ on the number of updates of $\Covered{v,i}$ in each
node-instance pair.  In 
total, we obtain the claimed worst-case bound on the running time.
\end{proof}

\section{Optimal oracle estimator  } \label{oracleLstar:sec}

We show that 
the distance-based influence estimator presented in 
Algorithm~\ref{TIoracle:alg}
is an instance of the \lstar\ estimator of \cite{sorder:PODC2014}. 

We first explain the derivation of the estimator. 
Similarly to the 
 single node, this is a sum estimator applied to 
  each summand $(v,i)$ in Equation~\eqref{relinf:eq} to obtain an 
  unbiased estimate of 
$\alpha(\min_{u\in 
    S}d^{(i)}_{uv})$. 
Rank values $r$ that belong to pairs $(v,i)$ where $v\in S$ can be 
identified because they must appear as $(r,0)$ in $\cADS(v)$ when we 
use structured permutation ranks. 
For 
these pairs, we do not compute an estimate but simply add an exact contribution of $|S|\alpha(0)$ to the 
estimated influence. 

  The remaining rank values are those associated with pairs that have 
  a positive distance from $S$.  That is,
a pair $(v,i)$ so that  $d^{(i)}_{Sv}>0$. 
We would like to estimate $\alpha(\min_{u\in 
    S}d^{(i)}_{uv})$ for each such pair $(v,i)$. 
We apply the 
 \lstar\ estimator of Cohen~\cite{sorder:PODC2014}, which is applicable 
 to any monotone estimation problem (MEP) and is the unique admissible 
 (nonnegative unbiased) monotone estimator. 
  We first show that our estimation problem is a MEP. 
 To represent the problem as a MEP, we fix the rank values of all 
  pairs, and consider the outcome (presence in $\cADS(u)$ for $u\in 
  S$) as a function of the ``random seed'' $r^{(i)}_v$.  Clearly,
the lower $r^{(i)}_v$ is, the more information we have (a tighter 
upper bound) on $\alpha(d^{(i)}_{Sv})$.  Therefore, the problem is a MEP. 

  The information in the outcome is actually contained in the skyline 
\Skylines{$r^{(i)}_v$}
 derived from the set of occurrences of the rank 
 $r^{(i)}_v$ in the sketches of $S$. 
When the skyline for a pair $(v,i)$  is empty ($r^{(i)}_v$ is not included in any of the sketches of 
  nodes in $S$), the 
 \lstar\ estimate is $0$ and is not explicitly computed. 

Note that the estimator is applied to each 
  rank value $r^{(i)}_v\not\in Z$ (and does not require explicit knowledge of the corresponding pair 
  $(v,i)$). 

Since the L$^*$ estimator is admissible, it dominates the union 
estimator and in particular has CV at most $1/\sqrt{2(k-1)}$. 
Note that this is an upper bound on the CV; the worst case is $|S|=1$
or when sketches are very similar.  As with reachability sketches, we can expect the estimate quality to 
be up to a factor of $\sqrt{|S|}$ smaller when the ``coverage'' of different 
seed nodes are sufficiently different.  Similarly, Chernoff 
concentration bounds apply here:  If we use $k=\epsilon^{-2} c\ln n $, the relative error 
exceeds $\epsilon$ with probability $\leq 1/n^c$.  Therefore,
with high probability, we can be accurate on polynomially many queries.

\section{Running time of $\alpha$-SKIM} \label{askimtime:proofs}

 This section provides the proof of  Theorem \ref{alphatime:thm}.
We also provide a slightly tighter worst-case bound on the (expected) running time for
  (i) the important cases of polynomial or exponential decay (both are
  instances of $\alpha$ with
  nonpositive relative rate of change) and (ii) general decay
  functions when we consider approximation with respect to relaxed influence
  computed with slight perturbations in distances:
\begin{theorem} \label{alphatimemodified:thm}
A modified Algorithm \alphaskim\ runs in expected time 
$$O(\frac{\log^2 n}{\epsilon}\sum_{i=1}^\ell |E^{(i)}| + \frac{\log^2
  n}{\epsilon^3}n+ \frac{\log n}{\epsilon^2}|\bigcup_{i=1}^\ell E^{(i)}|)= O(\epsilon^{-3}\sum_{i=1}^\ell |E^{(i)}| \log^2 n )$$
providing the guarantee \eqref{approxguarantee} when $(\ln \alpha(x))' \leq 0$  ($\alpha$ has nonpositive relative
rate of change).
For a general decay function $\alpha$, we obtain with high probability
approximation with respect to  $\max_{U\mid |U|=s}
\tilde{\INF}(\mathcal{G},U)$, where $\tilde{\INF}(\mathcal{G},S)\equiv
\sum_{u\in V} \alpha((1+\epsilon) d_{Su})$.
\end{theorem}

  We first list the main components of the \alphaskim\ computation, the bounds we
  obtain on the computation of each component, and pointers to the proofs in the sequel.
The analysis assumes that $\lambda=0.5$.

\begin{itemize}
\item
The reversed Dijkstra runs when building the sketches, including
updating the sample and estimation data structures as new entries are
inserted.  Since we use efficient structures to identify which runs need to be
resumed and to update estimates and samples, this component is dominated by
the Dijkstra computations.  We obtain a bound of 
$O(k \log^2 n \sum_v \max_{i\in [\ell]} deg^{(i)}(v))$ on the number 
of operations (see Lemma \ref{reversecost:lemma}). 
\item
The forward Dijkstra runs which update the residual
problem after a seed node is selected.   We express the bound in terms
of the maximum number of times, which we denote by $X$, that
$\delta^{(i)}_v$ is decreased  for a pair
$(i,v)$.  We consider both bounding $X$ and its impact on the time
bound In Section \ref{deltadecbound:sec}.
We show that the forward Dijkstra runs take  $O(X \sum_{i\in [\ell]}
  |E^{(i)}| \log n)$ operations (see Lemma \ref{fdijkstracost:lemma}).  
We obtain a bound of $X=O(\epsilon^{-1} \log^2 n)$ using Theorem \ref{randupdatebound:thm}.
We also obtain a slightly tighter
bound of $X = O(\epsilon^{-1}\log n)$ for a slightly modified
algorithm.  The modified algorithm has the same approximation
guarantee (as in Theorem \ref{alphaaccuracy:thm}) when $\alpha$ has
nonpositive relative rate of change.  For general $\alpha$, the
guarantee is with respect to a relaxed condition.  In practice on realistic inputs, however, we
expect $X$ to be small and we observed good running times even without
the modifications.
\item
Updating the sample and estimation structures when entries are
reclassified down.  This happens after a seed is added when the
forward Dijkstra run from the seed in instance $i$ updates
$\delta^{(i)}_v$.  This means all samples which include a
  $(v,i)$ entry need to be updated.  Expressed in terms of $X$, this
  cost is $O(X nk \log n)$ (see Lemma \ref{movedowncost:lemma}).
\item
Updating the sampling and estimation structures when entries are reclassified
up.  This can happen after $\tau$ is decreased.   We obtain a bound of
$O(nk\log^2 n)$  (in Section \ref{thressecbounb}).
\end{itemize}

\subsection{The threshold $\tau$} \label{thressecbounb}
An immediate upper bound on the maximum influence of a node
is $n \alpha(0)$.  This means that we can safely initialize our
algorithm with
$\tau = (n\alpha(0)\ell )/k$ and have an
expected  initial sample sizes $O(k)$ for each node.

We next observe that we can safely terminate the algorithm when $\tau$
decreases to below $\tau \leq 
\epsilon \alpha(0)\ell/k$ and incur at most $O(\epsilon)$ contribution
to the error.
To establish that, first observe
that the threshold $\tau$ is decreased only when the maximum estimated marginal 
influence, over all nodes, is less than $\tau k$.   Therefore,
when $\tau \leq
\epsilon \alpha(0)\ell/k$l,  the maximum marginal
influence of all nodes is at most $\epsilon \alpha(0)$.  
Now observe that at this point, for all $s' \geq 0$,  if 
we add $s'$ additional seed nodes, the relative contribution of these nodes
is at most
$\leq (s'\epsilon \alpha(0)/((s+s')\alpha(0)) \leq \epsilon$ (The
denominator $(s+s')\alpha(0)$ is a lower bound on the
total influence of any seed set of size $s+s'$).

 Combining the start and end values of $\tau$, we obtain the following bound on the total number of times $\tau$ is decreased:
\begin{lemma} \label{taubound:lemma}
$\tau$ can decrease at most $\log_{1/\lambda} (n/\epsilon) =O(\log n)$ times. 
\end{lemma}

 When $\tau$ decreases, we have the property 
that the  PPS samples of all nodes are of 
size smaller than $k$ (otherwise we have an estimate that exceeds
$k\tau$).  

After $\tau$ is decreased, we extend the samples to be with respect to 
the new $\tau$.  We  show that the 
expected sample size remains $O(k)$:
 \begin{lemma} \label{samplesize:lemma}
For every node, the expected number of entries after a threshold 
decrease $\tau \gets \lambda \tau$  is at most $k/\lambda$, with good concentration. 
 \end{lemma}
\begin{proof}
 Equivalently, we consider the size of a PPS sample with threshold 
$\lambda \tau$ when the total weight of the set is at most $k \tau$. 
\end{proof}

We are now ready to bound the total number of
reclassification of sample entries.

\begin{lemma} \label{samplemembers:lemma}
The total number of reclassifications of entries in the sample is 
$O(nk\log n)$
\end{lemma}
\begin{proof}
An entry in $\Index{v,i}$ can only be reclassified down 3 times before
it is removed from the sample (from H to M,  M to L, or L to
removal), unless it is reclassified up.  An entry can be reclassified
up only  when $\tau$
decreases, which happens at most  $O(\log n)$ times.  Each decrease 
``resets'' at most $kn$ entries to class $H$: Either existing entries
reclassified up or at most new entries (since by Lemma
\ref{samplesize:lemma} expected sample sizes remain $O(k)$).  These entries can then be reclassified down at most 
$3$ times before they are eliminated from the sample.  So the total
number is $O(nk\log n)$.
\end{proof}

  We are now ready to bound the total work performed by updating the
  sample and estimation components by entries being reclassified up as
  a result of a $\tau$ decrease ($\MoveUp$ calls).  The 
  cost of each such call is proportional to the number of 
  reclassified entries. It also requires a call to the priority queue
  to efficiently find all inverted samples with at least one
  reclassified element.  In the worst case, the cost is
  $O(\log(n\ell)=O(\log n)$ times the number of reclassifications.
 In total using Lemma \ref{samplemembers:lemma}, we obtain a worst
 case bound of $O(nk\log^2 n)$ on the reclassification-up component of
 the computation.

\subsection{Bounding the reversed Dijkstra computations}

 We bound the expected total number of distinct entries (pairs
 $(v,i)$) that
were included in the PPS sample of a node $u$  at any point during the execution of
the algorithm.  

\begin{lemma} \label{samplemembers2:lemma}
For a node $v$, the number of distinct entries in the sample of $v$
during the execution of the algorithm is $O(k\log n)$ with good
concentration (of the upper bound).
\end{lemma}
\begin{proof}
Each decrease of $\tau$ introduces in expectation $O(k)$ new entries,
and there are $O(\log n)$ such decreases (Lemma \ref{taubound:lemma}).
\end{proof}

  We can now bound the work of the reverse Dijkstra runs used to
  construct the sketches
\begin{lemma}  \label{reversecost:lemma}
The number of operations performed by the reverse Dijkstra runs is 
$O(k \log^2 n \sum_v \max_{i\in [\ell]} deg^{(i)}(v))$
\end{lemma}
\begin{proof}
Each productive scan of a node $u$ by a reverse Dijkstra sourced at
$(v,i)$ (productive means that
the node was next on the Dijkstra state priority queue) means that the
entry $(v,i)$ is inserted into a PPS sample of $u$, updating the estimation structure accordingly.  This involves $O(\log n)$ operations in updating
priority queues in the state of the Dijkstra run, structures
maintaining the samples, and looking at all 
outgoing edges of the node in the transposed  instance $G^{(i)}$.

Each such scan can be charged to an entry inserted into a PPS sample.
From Lemma \ref{samplemembers2:lemma}, we obtain that each node, 
in expectation, can have $O(k\log n)$ such entries.  Therefore, the
node is scanned $O(k\log n)$ times.   
\end{proof}
We remark that if the instances are generated by an IC model, we can
replace $\max_{i\in [\ell]} deg^{(i)}(v)$  by the expected degree
$\E[deg(v)]$ and accordingly obtain the bound $O(k \log^2 n \E[|E|])$.

\subsection{Bounding the expected number of times $\delta^{(i)}_v$
  decreases for a certain pair $(v,i)$} \label{deltadecbound:sec}
We now bound the n umber of updates of $\delta^{(i)}_v$ performed as
seed are added when maintaining the
residual problem.

If we have a bound of $X$ on the number of updates per node-instance
pair, then
\begin{lemma}  \label{fdijkstracost:lemma}
The computation of the forward 
Dijkstra runs is $O(X \sum_{i\in [\ell]} |E^{(i)}| \log n)$. 
\end{lemma}
\begin{proof}
Each node-instance scan 
can be charged to a decrease of $\delta^{(i)}_v$.   
\end{proof}

  We also can express the total cost of the 
$\MoveDown{}$ calls by $X$.
\begin{lemma}  \label{movedowncost:lemma}
The total computation of all $\MoveDown$ calls is $$O(X nk \log n)\ .$$
\end{lemma}
\begin{proof}
 Each call to $\MoveDown$ for $(v,i)$ updates a value for $(v,i)$ in
 a priority queue (at $O(\log n)$ cost), which is of the order of the
 forward Dijkstra computation that generated the update of
 $\delta^{(i)}_v$.
Otherwise, the $\MoveDown$ call performs a
 number of operations that is 
linear in the number of active entries in $\Index{v,i}$
(entries that are in a sample of some node).  
In addition, 
$\MoveDown$ may also permanently discard entries at the tail of the $\Index{v,i}$,
 but the removal of these entries is charged to their insertion.

 It remains to bound the computation of $\MoveDown$ when processing
 active entries of $\Index{v,i}$.
Using Lemma \ref{samplemembers2:lemma}, there is a total of $O(nk\log n)$ entries that were active
in a sample at any point during the execution.  Each such entry can be
affected at most $X$ times.
\end{proof}

  We now bound $X$.  As argued in Section~\ref{thresh:sec}, we expect
  $X=O(\log n)$ on realistic instances.  Now noting that our
  sampled-based greedy selection satisfies the condition of Theorem
  \ref{randupdatebound:thm},
we obtain a bound of $X= O(\epsilon^{-2}\log^2 n)$.

  Here we propose modifications of the algorithm that allow us to
  obtain a slightly tighter bound on $X$ in interesting cases.  The first case covers all smooth
  decay functions that are exponential or slower:
\begin{lemma}
We can
modify \askim\ so that 
when $\alpha(x)$ has a nonpositive relative rate of change, that is,  $(\ln \alpha(x))' \geq 0$, then
$$X=O(\epsilon^{-1} \log n)\ .$$ The modification preserves the approximation ratio
stated in Theorem \ref{alphaaccuracy:thm}.
\end{lemma}
\begin{proof}
  The requirement  $(\ln \alpha(x))' \geq 0$ implies that 
 for all $x\geq 0$, $d \geq 0$, and $\Delta \geq 0$,
$$\frac{\alpha(d-\Delta)-\alpha(d)}{\alpha(d)} \geq 
\frac{\alpha(x+d-\Delta)-\alpha(d)}{\alpha(x+d)}\ .$$
This means that when we apply the following prune rules on forward updates of 
$\delta^{(i)}_v$:
We prune  at nodes where 
\begin{equation}\label{prunerule}
\frac{\alpha(d-\Delta)-\alpha(d)}{\alpha(d)} \leq \epsilon\ ,
\end{equation}
 where $d$ is 
the current value of $\delta^{(i)}_v$ and $d-\Delta$ is the updated value,
the condition \eqref{prunerule}  would actually hold for all nodes in instance $i$ reachable from 
$v$ via the Dijkstra search (since all these nodes have larger 
$\delta^{(i)}_v$.  

 The prune condition implies that for $(v,i)$ and all nodes Dijkstra 
 would have reached from the pruned one, the updated influence contribution 
by the better (closer) coverage is at most $\epsilon$ times the 
previous value.  So with this pruning, the influence of the seed set 
is captured with relative error of at most $\epsilon$. 

  We also observe that we can also always prune the Dijkstra 
  computations when the distance satisfies $\alpha(d) \leq
  \alpha(0)/n^2 \leq
 \epsilon \alpha(0)/n $.   

\ignore{
only to 
places where $\alpha(d) \geq (1+\epsilon)\alpha(\delta^{(i)}_v)$, this 
property holds to all pairs reachable from this propagation. 
We lose at most an $\epsilon$ factor on the coverage and limit the 
number of $\delta$ decreases per pair to $O(\epsilon^{-1}\log_2 
\epsilon^{-1})$. 
}

Combining, it means that with the prune rules, the total number of updates of 
$\delta^{(i)}_v$ per node-instance pair is $O(\epsilon^{-1} \log n)$. 
\end{proof}

 \begin{lemma}
We can modify the algorithm so that for any general decay function
$\alpha$, $X = O(\epsilon^{-1} \log n)$.  With the modification, we
obtain that with high probability,
$$\INF(\mathcal{G},S) \geq (1-(1-1/s)^s -\epsilon) \max_{ U \mid |U|=|S|}
\tilde{\INF}(\mathcal{G},U)\ ,$$
where $\tilde{\INF}(\mathcal{G},S)\equiv \sum_{u\in V} \alpha((1+\epsilon) d_{Su})$.
 \end{lemma}
\begin{proof}
We can apply a similar prune rule in the forward Dijkstra runs
  which updates only when the decrease to distance is at least 
  $\epsilon$ times the current distance. 
This would give us a bound on the number of updates,  but a weaker approximation 
  guarantee that holds with respect to a softened influence function
$\sum_{u\in V} \alpha((1+\epsilon)d_{Su})$.
\end{proof}

\section{Pseudocode}\label{app:ps}

\subsection{Functions for Distance-Based GREEDY}
\label{app:ps:funcgreedy}

This appendix contains the pseudocode of functions for our distance-based version of {\sc Greedy} from Section~\ref{model:sec}.

\begin{function} \caption{{MargGain}($u$): Marginal influence of $u$ }
\KwIn{Residual instance $(\mathcal{G},\delta)$ and node $u$}
\KwOut{$\INF((\mathcal{G},\delta),u)$}
$I_u \gets 0$ \tcp*{sum of marginal contributions} 
\ForEach{instance $i$}{Run Dijkstra from $u$ in 
  $G^{(i)}$, during which\\  \ForEach{visited node $v$ at distance
    $d$}{\lIf{$\alpha(d)=0$ or $d \geq \cdelta{v,i}$}{Prune}\lElse{$\increase{I_u}{\alpha(d)-\alpha(\cdelta{v,i})}$}
}}
\Return{$I_u/\ell$}
\end{function}
\begin{function}\caption{AddSeed($u$): Update $\delta$ according to
    $u$}
\KwIn{Residual instance $(\mathcal{G},\delta)$ and node $u$}
\ForEach{instance $i$}{Run Dijkstra from $u$ in 
  $G^{(i)}$, during which\\ \ForEach{visited node $v$ at distance $d$}{\lIf{$\alpha(d)=0$ or $d \geq \cdelta{v,i}$}{Prune}\lElse{$\cdelta{v,i} \gets d$}
}}
\end{function}

\newpage
\subsection{Algorithms for Threshold Model}
\label{app:ps:alg}

This appendix contains the pseudocode of algorithms for threshold
influence maximization (Section~\ref{thresh:sec})

\begin{algorithm2e}[!h]
\caption{Threshold IM (\tskim)} \label{Tskim:alg}
\DontPrintSemicolon

\KwIn{Directed graphs $\{G^{(i)}\}$, threshold $T$, parameter~$k$}
\KwOut{Sequence of node and marginal influence pairs}

\BlankLine

\tcp{Initialization}
\lForAll{node/instance pairs $(u,i)$}{$\Covered{u,i} \leftarrow \infty$}
\lForAll{nodes~$v$}{$\Size{v} \leftarrow 0$}
$\Index \leftarrow$ hash map of node-instance pairs to nodes\;
$\SeedList \leftarrow \perp$\tcp*{List of seeds \& marg.\ influences}
$\Rank \leftarrow 0$\;

\BlankLine
shuffle the $n\ell$ node-instance pairs~$(u,i)$\;
\BlankLine

\tcp{Compute seed nodes}
\While{$|\SeedList| < n$}{

\While(\tcp*[f]{Build sketches}){$\Rank < n\ell$}{
	$\Rank \leftarrow \Rank + 1$\;
	$(u,i)\leftarrow$ $\Rank$-th pair in shuffled sequence\;
	\BlankLine
    \lIf(\tcp*{Pair $(u,i)$ is covered}){$\Covered{u,i} < \infty$}{\Skip} 
    run Dijkstra from~$u$ in reverse graph~$G^{(i)}$, during which\;
    \ForEach{scanned node~$v$ in distance $d$}{
    	\lIf(\tcp*[f]{Prune at depth~$T$}){$d > T$}{\Prune}
		\BlankLine
      $\Size{v} \leftarrow \Size{v} + 1$\;
      $\Index{u,i} \leftarrow \Index{u,i} \cup \{v\}$\; \BlankLine
      \If{$\Size{v} = k$}{
        $x \leftarrow v$\tcp*{Next seed node}
        abort sketch building\;
      }
    }
}

\BlankLine
\If{all nodes $u$ have $\Size{u}<k$ }{
	$x   \leftarrow \argmax_{u \in V}{\Size{u}}$\;
}

\BlankLine
$I_x \leftarrow 0$\tcp*{The coverage of~$x$}
\ForAll(\tcp*[f]{Residual problem}){instances~$i$}{
	run Dijkstra from~$x$ in forward graph~$G^{(i)}$, during which\;
	\ForEach{scanned node~$v$ in distance $d$}{
		\lIf{$\Covered{v,i} \leq d$ \Or $d > T$}{\Prune}
		\lIf{$\Covered{v,i} = \infty$}{$I_x \leftarrow I_x + 1$}
		$\Covered{v,i} \leftarrow d$\;
		\BlankLine
		\ForAll{nodes $w$ in $\Index{v,i}$}{
			$\Size{w} \leftarrow \Size{w} - 1$\;
		}
		$\Index[v,i] \leftarrow \bot$\tcp*{Erase $(v,i)$ from $\Index$}
	}
}
\BlankLine
\SeedList.\Append{$x$, $I_x/\ell$}\;
}
\Return{\SeedList}\;
\end{algorithm2e}

\newpage
\subsection{Algorithms for Distance-Based Influence Oracle} \label{timedoracle:ps:alg}
This appendix contains the pseudocode for our 
distance-based influence oracle (Section~\ref{timedQ:sec}).

\begin{algorithm2e}[h]
\caption{Distance-Based Influence Oracle}\label{TIoracle:alg}
\DontPrintSemicolon 

\KwIn{Seed set~$S$, function~$\alpha$, sketches~$\cADS(u)$ for~$u\in S$}
\KwOut{Estimated influence for~$S$}
\BlankLine 
\tcp{Remember ranks who have distance zero in at least one sketch with respect to~$S$}
$\Zeroranks \leftarrow$ empty set (e.g.~hash map) of ranks\;
\ForAll{nodes~$u \in S$}{
	\ForEach{entry~$(r,d) \in \cADS(u)$ with~$d = 0$}{
		$\Zeroranks.\Insert{$r$}$\;
	}
}

\BlankLine 

\tcp{Build for each appearing rank a set of threshold rank/influence pairs}
$\Skylines \leftarrow$ new hash map from rank to array of pairs\;
\ForAll{nodes~$u \in S$}{
	\Queue $\leftarrow$ new max-heap of~$k$ smallest rank values\;
	\ForEach{entry~$(r,d) \in \cADS(u)$}{
		\eIf{$|\Queue| < k$}{
			\lIf{$r \not\in \Zeroranks$}{$\Skylines{$r$}.\Append{$(1.0, \alpha(d)$}$}
			$\Queue.\Insert{$r$}$\;
		}{ 
			\If{$r \not\in \Zeroranks$}{$\Skylines{$r$}.\Append{$(\Queue.\Maxelement{}, \alpha(d)$}$}
				$\Queue.\Insert{$r$}$\;
				$\Queue.\Deletemax{$r$}$\;
		}
	}
}

\BlankLine 

\tcp{Eliminate dominated entries}
\ForAll{ranks~$r \in \Skylines$}{
	\tcp{Sort by threshold rank in decreasing order. Break ties by
        decreasing $\alpha$}
	$\Sort{$\Skylines{$r$}$}$\;
	$\alpha^* \leftarrow 0$\;
	\ForAll{thresh.~rank/infl.~pairs~$(\tau, \alpha) \in \Skylines{$r$}$}{
		\lIf{$\alpha < \alpha^*$}{$\Skylines{$r$}.\Erase{$(\tau,\alpha)$}$}
		\lElse{$\alpha^* \leftarrow \alpha$}
	}
}

\BlankLine 

\tcp{This calls the L\textsuperscript{*} estimator for each skyline}
\Return $|S|\cdot\alpha(0) + (1/\ell)\cdot\sum_{r \in \Skylines}\Lstar{$\Skylines{$r$}$}$\;

\end{algorithm2e}

\begin{algorithm2e}[h]
\caption{\lstar estimator applied to a sorted skyline}
\label{Lest:alg}
\DontPrintSemicolon 
\KwIn{A sorted skyline~$\Skylines{r}=\{(\tau_j,\alpha_j)\}$}
\KwOut{$\lstar(\Skylines{r})$}

\BlankLine 

$S \gets 0$;\,   $x \gets 0$\;
\For{$i = 1, \ldots, |\Skylines{r}|$}{
	$x \gets (\alpha_i-S)/\tau_i$\tcp*{Note that $x$ is overwritten}
	\lIf{$i < |\Skylines{r}|$}{$S \gets S + x\cdot(\tau_{i}-\tau_{i+1})$}
}
\Return $x$\;
\end{algorithm2e}

\begin{algorithm2e}[h]
\caption{Combine rank-distance lists}\label{ADScomb:alg}
\DontPrintSemicolon
\KwIn{Two rank-distance lists~$A_1$ and~$A_2$}
\KwOut{Combined all-distance sketch~$\Ads_c$}
\BlankLine
$\ADS_c \leftarrow$ new (empty) ADS\;
\tcp{Merge sketches by increasing distance, breaking ties by increasing rank}
$\Tempsketch \leftarrow \Merge{$A_1$, $A_2$}$\;
\BlankLine
$\Numzero \leftarrow 0$\;
\Queue $\leftarrow$ new max-heap of~$k$ smallest rank values\;
\BlankLine
\tcp{Handle entries with distance 0}
\ForEach{entry~$(r,d) \in \Tempsketch$ with~$d=0$}{
\lIf{$\Numzero < k$}{$\ADS_c(u).\Pushback{$(r,d)$}$}
$\Queue.\Insert{$r$}$\;
\lIf{$|\Queue| > k$}{$\Queue.\Deletemax{}$}
$\Numzero \leftarrow \Numzero + 1$\;
}
\BlankLine
\tcp{Handle the rest of the entries}
\ForEach{entry~$(r,d) \in \Tempsketch$ with~$d > 0$}{
	\If{$|\Queue| < k$ or $r < \Queue.\Maxelement{}$}{
		$\ADS_c(u).\Pushback{$(r,d)$}$\;
		$\Queue.\Insert{$r$}$\;
		\lIf{$|\Queue| > k$}{$\Queue.\Deletemax{}$}
	}
}
\BlankLine
\Return $\Ads_c$\;
\end{algorithm2e}

\clearpage

\subsection{Functions for $\alpha$-SKIM}
\label{app:ps:funcalpha}

This appendix contains the subroutines of  Algorithm~\ref{alphaskim:alg} (\alphaskim), our fast algorithm for distance-based influence maximization.

\begin{function}\caption{NextSeed()}
\KwOut{The node $u$ which maximizes $\est.H[u]+\tau \est.M[u]$, if 
  happy with estimate. 
Otherwise $\perp$.}
\SetKw{True}{true}
\SetKw{False}{false}
\While{\True}{
\lIf{max priority in $\Qcands < k\tau$}{\Return{$\perp$}}
\Else{
 Remove maximum priority $u$ from $\Qcands$\;
$\hat{I}_u \gets \est.H[u] + \tau \est.M[u]$\;
\If{$\hat{I}_u \geq k\tau$ and $\hat{I}_u \geq $  $\max$ in 
  $\Qcands$}{$I_u\gets \MargGain{u}$;\\ \lIf{$I_u \geq  (1-1/\sqrt{k})\hat{I}_u$}{\Return{$(u,\hat{I}_u)$}}\Else{Place $u$ with priority $I_u$ in 
    $\Qcands$\; \Return{$\perp$}}}
\Else{Place $u$ with priority $\hat{I}_u$ in $\Qcands$}
}}
\end{function}

\begin{function}\caption{MoveUp()  Update estimates after
    decrreasing $\tau$}
\DontPrintSemicolon
\ForEach{ $(v,i)$ in $\Qhml$ with priority $\geq \tau$}{ delete
  $(v,i)$ from $\Qhml$\;
\tcp{Process $\Index{v,i}$}
\If(\tcp*[f]{move entries from M/L to H}){$\HM{v,i}\not=\perp$}{
  \While{$\HM{v,i} < |\Index{v,i}|$ and 
    $(u,d) \gets \Index{v,i}[\HM{v,i}]$ satisfies 
    $(c\gets \alpha(d)-\alpha(\delta^{(i)}_v)) \geq \tau$}{
$\increase{\est.H[u]}{c}$ \\
 \If(\tcp*[f]{Entry was M}){$\ML{v,i} =\perp$ or $\ML{v,i} > \HM{v,i}$}{$\decrease{\est.M[u]}{1}$} 
  $\increase{\HM{v,i}}{1}$}
\If{$\ML{v,i}\not=\perp$ and $\ML{v,i} < \HM{v,i}$}{$\ML{v,i} \gets \HM{v,i}$}
\If{$\HM{v,i} \geq |\Index{v,i}|$}{$\HM{v,i}\gets\perp$;
  $\ML{v,i}\gets\perp$}
}
\If(\tcp*[f]{Move from L to M}){$\ML{v,i} \not= \perp$}
{
  \While{$\ML{v,i} < |\Index{v,i}|$ and 
    $(u,d) \gets \Index{v,i}[\ML{v,i}]$ satisfies 
    $\alpha(d)-\alpha(\delta^{(i)}_v) \geq
    r^{(i)}_v\tau$}{$\increase{\ML{v,i}}{1}$; $\increase{\est.M[u]}{1}$}
\lIf{$\ML{v,i} \geq |\Index{v,i}|$}{$\ML{v,i}\gets\perp$}
}
$\UpdateReclassThresh{v,i}$ \tcp*{update $\Qhml$}
}
\end{function}

\begin{function}\caption{UpdateReclassThresh()$(v,i)$}
\KwOut{Update priority of $(v,i)$  in $\Qhml$}
$c\gets 0$;\\
\If{$\HM{v,i}\not=\perp$}{$(u,d) \gets \Index{v,i}[\HM{v,i}]$;
  $c\gets \alpha(d)-\alpha(\delta^{(i)}_v)$}
\If{$\ML{v,i}\not=\perp$}{$(u,d) \gets \Index{v,i}[\ML{v,i}]$;
  $c\gets \max\{c, (\alpha(d)-\alpha(\delta^{(i)}_v))/r^{(i)}_v\}$}
\If{$c>0$}{update priority of $(v,i)$ in $\Qhml$ to $c$} 
\end{function}

\begin{function}\caption{MoveDown() $((v,i),\delta_0,\delta_t)$}
\KwOut{Update estimation
    components for $(v,i)$  when
    $\delta^{(i)}_v$ decreases from $\delta_0$ to $\delta_t$}
\DontPrintSemicolon
$j \gets 0$; $t \gets \perp$; $\HM{v,i}\gets \perp$; \\ 
$z \gets
|\Index{v,i}|-1$; \lIf{$\ML{v,i} \not= \perp$}{$z \gets\ML{v,i}$}
$\ML{v,i}\gets
\perp$ \\
\While{$j \leq z$}{
$(u,d) \gets \Index{v,i}[j]$;\\
 \If(\tcp*[f]{entry was H}){$\alpha(d)-\alpha(\delta_0) \geq \tau$}
 {$\decrease{\est.H[u]}{\alpha(d)-\alpha(\delta_0)}$;\\ 
  \If(\tcp*[f]{is H}){$\alpha(d)-\alpha(\delta_t) \geq
    \tau$}{$\increase{\est.H[u]}{\alpha(d)-\alpha(\delta_t)}$}
 \ElseIf(\tcp*[f]{is M}){$\alpha(d)-\alpha(\delta_t) \geq r^{(i)}_v 
    \tau$}{$\increase{\est.M[u]}{1}$; \lIf{$\HM{v,i}=\perp$}{$\HM{v,i}=j$}}
  \ElseIf(\tcp*[f]{truncate}){$\alpha(d)
      \leq
      \alpha(\delta_t)$}{\lIf{$t=\perp$}{$t=j$}}
  \Else(\tcp*[f]{is L}){\lIf{$\ML{v,i}=\perp$}{$\ML{v,i}
        \gets j$}}
 }
\ElseIf(\tcp*[f]{entry was M}){$\alpha(d)-\alpha(\delta_0) \geq r^{(i)}_v
  \tau$}
{
\If(\tcp*[f]{is M}){$\alpha(d)-\alpha(\delta_t) \geq
    r^{(i)}_v \tau$}{\lIf{$\HM{v,i}=\perp$}{$\HM{v,i} \gets j$}}
 \Else(\tcp*[f]{is not M}){$\decrease{\est.M[u]}{1}$;\\
  \eIf(\tcp*[f]{truncate}){$\alpha(d)\leq\alpha(\delta_t)$}{\lIf{$t=\perp$}{$t=j$}}(\tcp*[f]{is
    L}){\lIf{$\ML{v,i}=\perp$}{$\ML{v,i}\gets j$}}
}}
 $\increase{j}{1}$
}
\lIf{$t\not=\perp$}{truncate $\Index{v,i}$ from $t$ on.}
\Else(\tcp*[f]{clean tail}){$t\gets |\Index{v,i}|-1$; \\
\While{$t \geq 0$ and $(u,d) \gets \Index{v,i}[t]$
    has $\alpha(d) \leq \alpha(\delta_t)$}{$\decrease{t}{1}$}
truncate $\Index{v,i}$ at position $t+1$ on}
Remove pair $(v,i)$ from $\Qhml$\;
$\UpdateReclassThresh{v,i}$ \tcp*{Update $\Qhml$}
\end{function}

\end{document}